\documentclass{CSML}
\pdfoutput=1

\usepackage{lastpage}
\newcommand{\dOi}{13(3:8)2017}
\lmcsheading{}{1--\pageref{LastPage}}{}{}%
{Dec.~02, 2016}{Aug.~07, 2017}{}

\usepackage[colorlinks=true]{hyperref}
\hypersetup{hidelinks}
\usepackage{tikz}
\usepackage{amssymb}

\usepackage{xcolor,latexsym,amsmath,extarrows,alltt}
\usepackage{xspace}
\usepackage{booktabs}
\usepackage{mathtools}
\usepackage{enumitem}
\usepackage{stmaryrd}
\usepackage{microtype}
\usepackage{wrapfig}
\usepackage{multirow}
\usepackage{thmtools}
\usepackage{tabularx}

\theoremstyle{plain}\newtheorem{theorem}{Theorem}[section]
\theoremstyle{plain}\newtheorem{lemma}[theorem]{Lemma}
\theoremstyle{plain}\newtheorem{corollary}[theorem]{Corollary}
\theoremstyle{definition}\newtheorem{definition}[theorem]{Definition}
\theoremstyle{definition}\newtheorem{remark}[theorem]{Remark}
\theoremstyle{definition}\newtheorem{example}[theorem]{Example}

\declaretheoremstyle[
spaceabove=6pt, spacebelow=6pt,
numbered=yes,
sibling=theorem,
headfont=\normalfont\bfseries,
notefont=\mdseries, notebraces={(}{)},
bodyfont=\normalfont,
postheadspace=1em,
preheadhook=\endgraf\nobreak\noindent\hrulefill\vspace*{\dimexpr-6pt+\topsep\relax},
prefoothook=\vspace*{\dimexpr-8pt+\topsep\relax}\endgraf\nobreak\noindent\hrulefill
]{algostyle}
\declaretheorem[style=algostyle]{algorithm}

\newcommand{\timecomp}[1]{\ensuremath{\textrm{TIME}\left(#1\right)}}
\newcommand{\etime}[1]{\ensuremath{\textrm{E}^{#1}\textrm{TIME}}}
\newcommand{\exptime}[1]{\ensuremath{\textrm{EXP}^{#1}\textrm{TIME}}}
\newcommand{\elementary}{\ensuremath{\textrm{ELEMENTARY}}}
\newcommand{\ptime}{\ensuremath{\textrm{P}}}
\newcommand{\OO}{\mathcal{O}}

\newcommand{\nats}[0]{\ensuremath{\mathbb{N}}}
\newcommand{\A}{\mathcal{A}}
\newcommand{\B}{\mathcal{B}}
\newcommand{\F}{\mathcal{F}}
\renewcommand{\P}{\mathbb{P}}
\newcommand{\V}{\mathcal{V}}
\newcommand{\Flab}{\mathcal{F}_{\mathtt{lab}}}
\newcommand{\NF}{\mathcal{N}\!\mathcal{F}}
\newcommand{\Sorts}{\mathcal{S}}
\newcommand{\Constructors}{\mathcal{C}}
\newcommand{\Defineds}{\mathcal{D}}
\newcommand{\Data}{\mathcal{D}\!\!\mathcal{A}}
\newcommand{\Terms}{\mathcal{T}}
\newcommand{\Rules}{\mathcal{R}}
\newcommand{\Ruleslab}{\Rules_{\mathtt{lab}}}
\newcommand{\Var}{\mathit{Var}}
\newcommand{\asort}{\iota}

\newcommand{\atype}{\sigma}
\newcommand{\btype}{\tau}
\newcommand{\ctype}{\pi}

\newcommand{\apps}[3]{#1\ #2 \cdots #3}
\newcommand{\symb}[1]{\mathtt{#1}}
\newcommand{\identifier}[1]{\mathsf{#1}}
\newcommand{\unknown}[1]{\underline{\symb{#1}}}
\newcommand{\encode}[1]{\overline{\symb{#1}}}
\newcommand{\interpret}[1]{\llbracket #1 \rrbracket_{\B}}
\newcommand{\numinterpret}[1]{\langle #1 \rangle}
\newcommand{\labl}{\mathsf{label}}
\newcommand{\numrep}[1]{[#1]}
\newcommand{\Conf}{\mathsf{Confirmed}}
\newcommand{\arrtype}{\Rightarrow}
\newcommand{\arrz}{\to}
\newcommand{\arrzt}{\Rightarrow}
\newcommand{\arr}[1]{\arrz_{#1}}
\newcommand{\arrr}[1]{\arr{#1}^*}
\newcommand{\arrp}[1]{\arr{#1}^+}
\newcommand{\supterm}{\rhd}

\newcommand{\suptermeq}{\unrhd}
\newcommand{\transition}[5]{#1~\displaystyle{\mathop{=\!\!=\!\!\!
  \Longrightarrow}^{#2/#3\ #4}}~#5}
\newcommand{\card}{\mathtt{card}}
\newcommand{\blank}{\textbf{\textvisiblespace}}
\newcommand{\nul}{\symb{0}}
\newcommand{\one}{\symb{1}}
\newcommand{\nil}{\symb{[]}}
\newcommand{\cons}{\symb{;}}
\newcommand{\q}{\symb{?}}
\newcommand{\h}{\symb{\#}}
\newcommand{\strue}{\symb{true}}
\newcommand{\sfalse}{\symb{false}}
\newcommand{\bits}{\symb{list}}
\newcommand{\bit}{\symb{symb}}
\newcommand{\symbs}{\symb{symb}}
\newcommand{\bool}{\symb{bool}}

\newcommand{\secshort}{\S~}

\begin{document}

%\thispagestyle{empty}
%\definecolor{comment}{rgb}{0.92, 0.92, 0.92}
%\colorbox{comment}{\parbox{13cm}{
%  This paper is an extended version of a paper presented at DICE
%  2016, including complete proofs.  The 5-paper extended
%  abstract in the DICE proceedings can be found at:
%  \begin{center}
%  \url{http://cl-informatik.uibk.ac.at/users/kop/dice16.pdf}
%  \end{center}
%}}
%\setcounter{page}{0}
%
%\newpage

\title{Complexity Hierarchies and Higher-Order Cons-Free Term Rewriting}
\thanks{Supported by the Marie Sk{\l}odowska-Curie
action ``HORIP'', program H2020-MSCA-IF-2014, 658162 and by the
Danish Council for Independent Research Sapere Aude grant ``Complexity
via Logic and Algebra''}% (COLA)}
\author[C.~Kop]{Cynthia Kop} 
\author[J.G.~Simonsen]{Jakob Grue Simonsen} 
\address{Department of Computer Science, Copenhagen University}
\email{\{kop,simonsen\}@di.ku.dk}
\maketitle

\begin{abstract}
Constructor rewriting systems are said to be cons-free if, roughly,
constructor terms in the right-hand sides of rules are subterms of
the left-hand sides; the computational intuition
is that rules cannot build new data structures.
In programming language research, cons-free languages have been used
to characterize hierarchies of computational complexity classes; in
term rewriting, cons-free first-order TRSs have been used
%to characterize the polynomial-time decidable sets.
to characterize $\ptime$.

We investigate cons-free higher-order term rewriting systems, the
complexity classes they characterize, and how these depend on the
type order of the systems. We prove that, for every $K
\geq 1$, left-linear cons-free systems with type order $K$
characterize $\etime{K}$ if unrestricted evaluation is used
(i.e., the system does not have a fixed reduction strategy).

The main difference with prior work in implicit complexity is that
(i) our results hold for non-orthogonal TRSs
%with possible rule overlaps
with no assumptions on reduction strategy,
(ii) we consequently obtain much larger classes for each
type order ($\etime{K}$ versus $\exptime{K-1}$), and
(iii) results for cons-free term rewriting systems have previously
only been obtained for $K=1$, and with additional syntactic
restrictions besides cons-freeness and left-linearity.

Our results are %apparently
among the first implicit
characterizations of the hierarchy $\textrm{E} = \etime{1} \subsetneq
\etime{2} \subsetneq \cdots$.
Our work confirms prior results that having full non-determinism (via
overlapping rules) does not directly allow for characterization of
non-deterministic complexity classes like $\textrm{NE}$.  We also show
that non-determinism makes the classes characterized highly sensitive to minor syntactic changes like admitting product types or
non-left-linear rules.
\end{abstract}

\section{Introduction}

In~\cite{jon:01}, Jones introduces \emph{cons-free programming}:
working with a small functional programming language, cons-free
programs are exactly those where function bodies cannot contain use of
data constructors (the ``cons'' operator on lists).  Put differently,
a cons-free program is \emph{read-only}: data structures cannot be
created or altered, only read from the input; and any data passed as
arguments to recursive function calls must thus be part of the
original input.

The interest in such programs lies in their applicability to
computational complexity: by imposing cons-freeness, the resulting
programs can only decide the sets in a proper subclass of the
Turing-decidable sets; indeed are said to \emph{characterize} the
subclass.  
Jones shows that adding further restrictions such as type order
or enforcing tail recursion lowers the resulting expressiveness to known
classes.  For example, cons-free programs with data order $0$ can
decide exactly the sets in PTIME, while tail-recursive cons-free
programs with data order $1$ can decide exactly the sets in PSPACE.
The study of such restrictions and the complexity classes characterized
is a research area known as \emph{implicit complexity}
and has a long history with many distinct approaches
(see, e.g., \cite{DBLP:conf/tlca/Baillot07,baillot_et_al:LIPIcs:2012:3664,DBLP:conf/esop/BaillotGM10,DBLP:journals/cc/BellantoniC92,DBLP:journals/apal/BellantoniNS00,Hofmann:hab,DBLP:journals/tcs/KristiansenN04}).

Rather than a toy language, it 
is tantalizing to consider \emph{term rewriting} instead.  Term
rewriting systems have no fixed evaluation order (so call-by-name or
call-by-value can be introduced as needed, but are not
\emph{required});
and term rewriting is natively non-deterministic,
allowing distinct rules to be applied (``functions to be invoked'')
to the same piece of syntax,
hence could be useful for extensions towards
non-deterministic complexity classes. Implicit complexity using
term rewriting has seen significant advances using a plethora
of approaches (e.g.\ \cite{DBLP:conf/aplas/AvanziniEM12,DBLP:conf/rta/AvanziniM10,DBLP:journals/corr/AvanziniM13}).
Most of this research has, however, considered fixed evaluation orders (most prominently innermost reduction),
and if not, then systems which are either orthogonal, or at least confluent (e.g. \cite{DBLP:conf/rta/AvanziniM10}).
Almost all of the work considers only first-order rewriting.

The authors of~\cite{car:sim:14} provide a first definition of
cons-free term rewriting without constraints on evaluation order or
confluence requirements, and prove that this class---limited to
\emph{first-order} rewriting---characterizes $\textrm{PTIME}$.
However, they impose a rather severe partial linearity restriction on
the programs.  
This paper seeks to answer two questions: (i) what happens if
\emph{no} restrictions beyond left-linearity and cons-freeness are
imposed? And (ii) what if we consider \emph{higher-order} term
rewriting?
We obtain that $K^{\text{th}}$-order cons-free term rewriting exactly
characterizes $\etime{K}$.  This is surprising because in Jones'
rewriting-like language, $K^{\text{th}}$-order programs characterize
$\exptime{K-1}$:
surrendering both determinism and evaluation order thus significantly
increases expressivity.
Our results are comparable to work in descriptive complexity theory
(roughly, the study of logics characterizing complexity classes) where the \emph{non}-deterministic
classes $\textrm{NEXP}^{K-1}\textrm{TIME}$ in the exponential hierarchy are exactly the sets axiomatizable by $\Sigma_K$ formulas in appropriate query logics \cite{KUPER199333,HELLA2006197}.

\section{Preliminaries}

\subsection{Computational Complexity}

We presuppose introductory working knowledge of computability and
complexity theory (%corresponding to standard textbooks,
see, e.g., \cite{Jones:CompComp}). Notation is fixed below.

Turing Machines (TMs) are tuples $(I,A,S,T)$ where $I \supseteq \{
0,1 \}$ is a finite set of \emph{initial symbols}; $A \supseteq I \cup
\{\blank\}$ is a finite set of \emph{tape symbols} with $\blank \notin
I$ the special \emph{blank} symbol; $S \supseteq \{ \symb{start},
\symb{accept}, \symb{reject} \}$ is a finite set of states; and $T$ is
a finite set of transitions $(i,r,w,d,j)$ with $i \in S \setminus
\{\symb{accept},\symb{reject}\}$ (the \emph{original state}), $r \in
A$ (the \emph{read symbol}), $w \in A$ (the \emph{written symbol}), $d
\in \{ \symb{L},\symb{R} \}$ (the \emph{direction}), and $j \in S$
(the \emph{result state}).
We also write this transition as $\transition{i}{r}{w}{d}{j}$.
All machines in this paper are \emph{deterministic}: every pair
$(i,r)$ with $i \in S \setminus \{ \symb{accept},\symb{reject} \}$ is
associated with exactly one transition $(i,r,w,d,j)$.
Every Turing Machine in this paper has a single, right-infinite tape.

A \emph{valid tape} is a right-infinite sequence of tape symbols with
only finitely many not $\blank$.  A \emph{configuration} of a TM is a
triple $(t,p,s)$ with $t$ a valid tape, $p \in \nats$ and $s \in S$.
The transitions $T$ induce a binary relation $\arrzt$ between
configurations in the obvious way.

\begin{definition}
Let $I \supseteq \{\nul,\one \}$ be a set of \emph{symbols}.
A \emph{decision problem} is a set $X \subseteq I^+$.
\end{definition}

A TM with input alphabet $I$ \emph{decides} $X \subseteq I^+$ if for
any string $x \in I^+$, we have $x \in X$ if{f} $(\blank x_1\dots
x_n\blank\blank\dots,0,\symb{start}) \arrzt^* (t,i,\symb{accept})$ for
some $t,i$, and $(\blank x_1\dots x_n\blank \blank\dots,0,
\symb{start}) \arrzt^* (t,i,\symb{reject})$ otherwise (i.e., the
machine halts on all inputs, ending in $\symb{accept}$ or
$\symb{reject}$ depending on whether $x \in X$).
If $f : \nats \longrightarrow \nats$ is a function, a (deterministic)
TM \emph{runs in time} $\lambda n.f(n)$ if, for each $n \in \nats
\setminus \{0\}$ and each $x \in I^n$, we have
$(\blank x\blank\blank\dots,0,\symb{start}) \arrzt^{\leq f(n)}
(t,i,\unknown{s})$ for some $\unknown{s} \in \{\symb{accept},
\symb{reject}\}$, where $\arrzt^{\leq f(n)}$ denotes a sequence of at
most $f(n)$ transitions.

\medskip
We categorize decision problems into classes based on the
time needed to decide them.

\begin{definition}
Let $f : \nats \longrightarrow \nats$ be a function. Then,
$\timecomp{f(n)}$ is the set of all $S \subseteq I^+$ such that
there exist  $a > 0$ and a deterministic TM running in time $\lambda
n.a \cdot f(n)$ that decides $S$ (i.e., $S$ is decidable in time
$\OO(f(n))$).
Note that by design, $\timecomp{\cdot}$ is closed under $\OO$.
\end{definition}

\begin{definition}
For $K,n \geq 0$, let $\mathrm{exp}_2^0(n) = n$ and
$\mathrm{exp}_2^{K+1}(n) = 2^{\mathrm{exp}_2^K(n)} =
\mathrm{exp}_2^K(2^n)$.

For $K \geq 1$ define:
$
\etime{K} \triangleq \bigcup_{a \in \nats}
\timecomp{\textrm{exp}_2^{K}(an)}
$.
\end{definition}

Observe in particular that $\etime{1} = \bigcup_{a \in \nats}
\timecomp{\textrm{exp}_2^{1}(an)} = \bigcup_{a \in \nats}
\timecomp{2^{an}} = \textrm{E}$ (where $\textrm{E}$ is the usual
complexity class of this name, see e.g.,
\cite[Ch.\ 20]{Papadimitriou:complexity}).
Note also that for any $d,K \geq 1$, we have $(\mathrm{exp}_ 2^K(x))^d
= 2^{d \cdot \mathrm{exp}_2^{K-1}(x)} \leq 2^{\mathrm{exp}_2^{K-1}(
dx)} = \mathrm{exp}_2^K(dx)$.  Hence, if $P$ is a polynomial with
non-negative integer coefficients and the set $S \subseteq \{0,1\}^+$
is decided by an algorithm running in $\timecomp{P(\mathrm{exp
}_ 2^K(an))}$ for some $a \in \nats$, then $S \in \etime{K}$.

By the Time Hierarchy Theorem \cite{Sipser:comp},
$\textrm{E} = \etime{1} \subsetneq \etime{2} \subsetneq \etime{3}
\subsetneq \cdots$.  The union
$\bigcup_{K \in \nats} \etime{K}$ is the set $\elementary$ of
elementary-time computable languages.

We will also sometimes refer to
$
\exptime{K} \triangleq \bigcup_{a,b \in \nats}
\timecomp{\textrm{exp}_2^{K}(an^b)}
$.

\subsection{Applicative term rewriting systems}\label{sec:atrss}

Unlike first-order term rewriting, there is no single, unified
approach to higher-order term rewriting, but rather a number of
different co-extensive systems with distinct syntax; for an overview
of basic issues, see \cite{T03_R}.
For the present paper, we have chosen to employ \emph{applicative
TRSs with simple types}, as (a) the applicative style and absence of
explicitly bound variables allows us to present our examples---in
particular the ``counting modules'' of
\secshort\ref{sec:counting}---in the most intuitive way, and (b) this
particular variant of higher-order rewriting is syntactically similar
to Jones' original definition using functional programming.
However, our proofs do not use any features of ATRS that preclude
using different formalisms; for a presentation using simply-typed
rewriting with explicit binders, we refer to the conference version
of this paper~\cite{kop:sim:16}.

%\paragraaf{Types and Terms}

\begin{definition}[Simple types]\label{def:types}
We assume given a non-empty set $\Sorts$ of \emph{sorts}.
Every $\asort \in \Sorts$ is a type of order \emph{0}.
If $\atype,\btype$ are types of order $n$ and $m$ respectively,
then $\atype \arrtype \btype$ is a type of order $\max(n+1,m)$.
Here $\arrtype$ is right-associative, so $\atype \arrtype \btype
\arrtype \ctype$ should be read $\atype \arrtype (\btype \arrtype
\ctype)$.
\end{definition}

We additionally assume given disjoint sets $\F$ of \emph{function
symbols} and $\V$ of variables, each equipped with a type.  This
typing imposes a restriction on the
formation of \emph{terms}:

\begin{definition}[Terms]\label{def:terms}
The set $\Terms(\F,\V)$ of terms over $\F$ and $\V$ consists of
those expressions $s$ such that $s : \atype$ can be derived for some
type $\atype$ using the following clauses:
(a) $a : \atype$ for $(a : \atype) \in \F \cup \V$, and
(b) $s\ t : \btype$ if $s : \atype \arrtype \btype$ and $t : \atype$.
\end{definition}

Clearly, each term has a \emph{unique} type.
A term \emph{has base type} if its type is in $\Sorts$, and \emph{has
functional type} otherwise.
We denote $\Var(s)$ for the set of variables occurring in a term $s$
and say $s$ is \emph{ground} if $\Var(s) = \emptyset$.
Application is left-associative, %i.e., $s\ t\ u$ denotes $(s\ t)\ u$;
so every term may be denoted $\apps{a}{s_1}{s_n}$ with $a \in \F
\cup \V$.  We call $a$ the \emph{head} of this term.
We will sometimes employ vector notation, denoting
$\apps{a}{s_1}{s_n}$ simply as $a\ \vec{s}$ when no confusion can
arise.

\begin{example}\label{ex:bitlists}
We will often use extensions of the signature
$\F_{\bits}$, given by:
\[
\begin{array}{rclcrclcrclcrcl}
\nul & : & \bit & \quad &
\one & : & \bit & \quad &
\nil & : & \bits & \quad &
\cons & : & \bit \arrtype \bits \arrtype \bits \\
\end{array}
\]
Terms are for instance $\one : \bit$ and
$\cons\ \nul\ (\cons\ \one\ \nil) : \bits$, as well as
$(\cons\ \nul) : \bits \arrtype \bits$.
However, we will always denote $\cons$ in a right-associative infix
way and only use it fully applied; thus, the second of these terms
will be denoted $\nul\cons\one\cons\nil$ and the third will not occur.
Later extensions of the signature will often use additional
constants of type $\symb{symb}$.
\end{example}

The notion of substitution from first-order rewriting extends in the
obvious way to applicative rewriting, but we must take special care
when defining subterms.

\begin{definition}[Substitution, subterms and contexts]
A \emph{substitution} is a type-preserving map from $\V$ to
$\Terms(\F,\V)$ that is the identity on all but finitely many
variables.  Substitutions
$\gamma$ are extended to arbitrary terms $s$, notation $s\gamma$,
by replacing each variable $x$ by $\gamma(x)$.
The \emph{domain} of a substitution $\gamma$ is the set
consisting of those variables $x$ such that $\gamma(x) \neq x$.

We say $t$ is a subterm of $s$, notation $s \suptermeq t$, if (a)
$s = t$, or (b) $s \supterm t$, where $s_1\ s_2 \supterm t$ if
$s_1 \supterm t$ or $s_2 \suptermeq t$.  In case (b), we say $t$ is a
\emph{strict} subterm of $s$.
%A \emph{context} $C[]$ is a term in $\Terms(\F,\V)$ in which a single
%occurrence of a variable has been replaced by a symbol $\Box \notin \F
%\cup \V$.  The result of replacing $\Box$ in $C[]$ by a term $s$
%(of matching type) is denoted $C[s]$.
\end{definition}

Note that $s_1$ is \emph{not} considered a subterm of $s_1\ s_2$; thus,
in a term $\apps{\identifier{f}}{x_1}{x_n}$ the only
strict subterms are $x_1,\dots,x_n$; the term $\apps{\identifier{
f}}{x_1}{x_{n-1}}$ (for instance) is \emph{not} a subterm.
The reason for this arguably unusual definition is that the
restrictions on rules we will employ do not allow us to ever isolate
the head of an application.  Therefore, such ``subterms'' would not be
used, and are moreover problematic to consider due to their higher type
order.
%Note also that in a context, the special symbol $\Box$ is allowed to
%occur at the head of an application.

\begin{example}\label{ex:introsucc}
Let $\symb{succ} : \bits \arrtype \bits$ be added to %the signature
$\F_{\symb{bits}}$ of Example~\ref{ex:bitlists}.
Then $\symb{succ}\ (\nul\cons\one\cons\nil) \supterm \one\cons\nil$,
but not $\symb{succ}\ (\nul\cons\one\cons\nil) \supterm
\symb{succ}$.
%Example contexts are $C_1 := \symb{succ}\ \Box$ and
%$C_2 := \Box\ (\one\cons\nil)$.  We have $C_1[\one\cons\nil] =
%C_2[\symb{succ}] = \symb{succ}\ (\one\cons\nil)$.
An example substitution is $\gamma := [xs:=y\cons\one\cons zs]$ (which
is the identity on all variables but $xs$), and for $s =
\symb{succ}\ (\nul\cons xs)$ we have $s\gamma = \symb{succ}\ 
(\nul\cons y\cons\one\cons zs)$.
\end{example}

At last we are prepared to define the reduction relation.

\begin{definition}[Rules and rewriting]\label{def:rules}
A \emph{rule} is a pair $\ell \arrz r$ of terms in $\Terms(\F,\V)$
with the same \emph{type}
%(i.e.~%there is $\asort \in \Sorts$ with
%$\ell : \atype$ and $r : \atype$ for some type $\atype$)
such that $\Var(r) \subseteq \Var(\ell)$.
%, and $\ell$ has no subterms $\apps{F}{
%s_1}{s_n}$ with $n > 0$ and $F \in \V$.
A rule $\ell \arrz r$ is \emph{left-linear} if every variable occurs
at most once in $\ell$.
Given a set $\Rules$ of rules, the reduction relation
$\arr{\Rules}$ on $\Terms(\F,\V)$ is given by:
\begin{center}
%\begin{tabular}{rcll}
%$C[\ell\gamma]$ & $\arr{\Rules}$ & $C[r\gamma]$ &
%  for any $\ell \arrz r \in \Rules$, context $C$, and substitution
%  $\gamma$ \\
%\end{tabular}
\begin{tabular}{rcll}
$\ell\gamma$ & $\arr{\Rules}$ & $r\gamma$ &
  for any $\ell \arrz r \in \Rules$ and substitution $\gamma$ \\
$s\ t$ & $\arr{\Rules}$ & $s'\ t$ & if $s \arr{\Rules} s'$ \\
$s\ t$ & $\arr{\Rules}$ & $s\ t'$ & if $t \arr{\Rules} t'$ \\
\end{tabular}
\end{center}
\end{definition}

Let $\arrp{\Rules}$ denote the transitive closure of $\arr{\Rules}$
and $\arrr{\Rules}$ the transitive-reflexive closure.
We say that \emph{$s$ reduces to $t$} if $s \arrr{\Rules} t$.
A term $s$ is in \emph{normal form} if there is no $t$ such that $s
\arr{\Rules} t$, and \emph{$t$ is a normal form of $s$} if $s
\arrr{\Rules} t$ and $t$ is in normal form.
An \emph{applicative term rewriting system}, abbreviated \emph{ATRS}
is a pair $(\F,\Rules)$
%, generating a set of terms and a reduction relation,
and its \emph{type order} (or just \emph{order}) is the maximal order
of any type declaration in $\F$.

\begin{example}\label{ex:plusfo}
Let $\F_{\symb{count}} = \F_\bits \cup \{ \symb{succ} : \bits \arrtype
\bits \}$ be the signature from Example~\ref{ex:introsucc}.
We consider the ATRS $(\F_{\symb{count}},\Rules_{\symb{count}})$ with
the following rules:
\[
\begin{array}{crclcrcl}
\mathsf{(A)} & \symb{succ}\ \nil & \arrz & \one\cons\nil \quad &
\mathsf{(B)} & \symb{succ}\ (\nul\cons xs) & \arrz & \one\cons xs \\
& & & &
\mathsf{(C)} & \symb{succ}\ (\one\cons xs) & \arrz & \nul\cons
  (\symb{succ}\ xs) \\
\end{array}
\]
This is a \emph{first-order} ATRS, implementing the successor function
on a binary number expressed as a bit string with the least significant
digit first.  For example, $5$ is represented by $\one\cons\nul\cons
\one\cons\nil$, and indeed $\symb{succ}\ (\one\cons\nul\cons\one\cons
\nil) \arr{\Rules} \nul\cons(\symb{succ}\ (\nul\cons\one\cons\nil))
\arr{\Rules} \nul\cons\one\cons\one\cons\nil$, which represents $6$.
\end{example}

\begin{example}\label{ex:plusho}
We may also define counting as an operation on \emph{functions}.
We let $\F_{\symb{hocount}}$ contain a number of typed
symbols, including $\nul,\one : \bit$,\ $\symb{o} : \symb{nat}$ and
$\symb{s} : \symb{nat} \arrtype \symb{nat}$ as well as
$\symb{set} : (\symb{nat} \arrtype \symb{symb}) \arrtype \symb{nat}
\arrtype \symb{symb} \arrtype \symb{nat} \arrtype \symb{symb}$.
%Consider $\F_{\symb{hocount}}$:
%\[
%\begin{array}{rclcrcl}
%\nul & : & \bit & \quad &
%\symb{o} & : & \symb{nat} \\
%\one & : & \bit & \quad &
%\symb{s} & : & \symb{nat} \arrtype \symb{nat} \\
%\symb{ifeq} & : & \multicolumn{5}{l}{
%  \symb{nat} \arrtype \symb{nat} \arrtype \symb{symb} \arrtype
%  \symb{symb} \arrtype \symb{symb}} \\
%\symb{fsucc} & : & \multicolumn{5}{l}{
%  (\symb{nat} \arrtype \bit) \arrtype \symb{nat}
%  \arrtype \symb{nat} \arrtype \bit} \\
%\symb{flip} & : & \multicolumn{5}{l}{
%  (\symb{nat} \arrtype \bit) \arrtype \symb{nat}
%  \arrtype \symb{nat} \arrtype \bit} \\
%\symb{check} & : & \multicolumn{5}{l}{
%  \bit \arrtype (\symb{nat} \arrtype \bit) \arrtype
%  \symb{nat} \arrtype \symb{nat} \arrtype \bit} \\
%\symb{set} & : & \multicolumn{5}{l}{
%  (\symb{nat} \arrtype \symb{symb}) \arrtype \symb{nat} \arrtype
%  \symb{symb} \arrtype \symb{nat} \arrtype \symb{symb}
%  } \\
%\end{array}
%\]
This is a second-order signature with \emph{unary} numbers $\symb{o},
\symb{s}\ \symb{o}, \symb{s}\ (\symb{s}\ \symb{o}), \dots$,
which allows us to represent the
bit strings from before as functions in $\symb{nat} \arrtype \bit$: 
a bit string $b_0\dots b_{n-1}$ corresponds to a function which
reduces $\symb{s}^i\ \symb{o}$ to $b_i$ for $0 \leq i < n$ and to
$\nul$ for $i \geq n$.
Let $\Rules_{\symb{hocount}}$ consist of the rules below;
types can be derived from context.  The
successor of a ``bit string'' $F$ is given by $\symb{fsucc}\ F\ 
\symb{o}$.
\[
\begin{array}{crclcrcl}
(\mathsf{D}) & \symb{ifeq}\ \symb{o}\ \symb{o}\ x\ y & \arrz & x &
(\mathsf{M}) & \symb{neg}\ \ & \arrz & \one \\
(\mathsf{E}) & \symb{ifeq}\ (\symb{s}\ n)\ \symb{o}\ x\ y & \arrz & y &
(\mathsf{N}) & \symb{neg}\ \one & \arrz & \nul \\
(\mathsf{F}) & \symb{ifeq}\ \symb{o}\ (\symb{s}\ m)\ x\ y & \arrz & y &
(\mathsf{O}) & \symb{nul}\ n & \arrz & \nul \\
(\mathsf{G}) & \symb{ifeq}\ (\symb{s}\ n)\ (\symb{s}\ m)\ x\ y &
  \arrz & \symb{ifeq}\ n\ m\ x\ y \\
(\mathsf{H}) & \symb{set}\ F\ n\ x\ m & \arrz &
  \symb{ifeq}\ n\ m\ x\ (F\ m) \\
(\mathsf{I}) & \symb{flip}\ F\ n & \arrz & \symb{set}\ F\ n\ 
  (\symb{neg}\ (F\ n)) \\
(\mathsf{J}) & \symb{fsucc}\ F\ n & \arrz & \symb{fsucchelp}\ (F\ n)\ 
  (\symb{flip}\ F\ n)\ n \\
(\mathsf{K}) & \symb{fsucchelp}\ \nul\ F\ n & \arrz & F \\
(\mathsf{L}) & \symb{fsucchelp}\ \one\ F\ n & \arrz & \symb{fsucc}\ F\ 
  (\symb{s}\ n) \\
\end{array}
\]
Rules (\textsf{I})--(\textsf{L}) have a functional type
$\symb{nat} \arrtype \symb{symb}$.
The function $\symb{nul}$ represents bit strings $0\dots 0$, and
if $F$ represents $b_0\dots b_{n-1}$ then %the function
$\symb{set}\ F\ (\symb{s}^i\ \symb{o})\ x$ represents
$b_0\dots b_{i-1} x b_{i+1}\dots b_{n-1}$.
%There are various representations of each number; 5
The number $5$ %(with bit string $101$)
is for instance
represented by $t := \symb{set}\ (\symb{set}\ \symb{nul}\ \symb{o}\ 
\one)\ (\symb{s}^2\ \symb{o})\ \one$.
We easily see that (**) $t\ \symb{o} \arrr{\Rules} \one$ and $t\ 
(\symb{s}\ \symb{o}) \arrr{\Rules} \nul$.
Intuitively, $\symb{fsucc}$ operates on $\one \dots \one\nul b_{i+1}
\dots b_{n-1}$ by flipping bits until some $\nul$ is encountered,
giving $\nul\dots\nul\one b_{i+1}\dots b_{n-1}$.
Using (**), $\symb{fsucc}\ t\ \symb{o} \arr{\Rules}
\symb{fsucchelp}\ (t\ \symb{o})\ (\symb{flip}\ t\ \symb{o})\ \symb{o}
\arrr{\Rules}
\symb{fsucchelp}\ \one\ (\symb{set}\ t\ \symb{o}\ (\symb{neg}\ \one))\ 
\symb{o} \arrr{\Rules}
\symb{fsucc}\ (\symb{set}\ t\ \symb{o}\ \nul)\ (\symb{s}\ \symb{o})
\linebreak\arrr{\Rules}
\symb{fsucchelp}\ \nul\ (\symb{set}\ (\symb{set}\ t\ \symb{o}\ \nul)\ 
(\symb{s}\ \symb{o})\ \one)\ (\symb{s}\ \symb{o})
\arr{\Rules}
\symb{set}\ (\symb{set}\ t\ \symb{o}\ \nul)\ (\symb{s}\ \symb{o})\ 
\one$; writing $u$ for this term, we can confirm that $u\ 
(\symb{s}^i\ \symb{o}) \arrr{\Rules} \one$ if only if $i=1$ or $i=2$:
$u$ represents $6$.
\end{example}

%\CKchange{(Example~\ref{ex:plusho} is important because
%binary counting will play a major role in this work.)}

For the problems we will consider, a key notion is that of \emph{data
terms}.

\begin{definition}\label{def:data}
We fix a partitioning of $\F$ into two disjoint sets, $\Defineds$ of
\emph{defined symbols} and $\Constructors$ of \emph{constructor
symbols}, such that $\identifier{f} \in \Defineds$ for all
$\identifier{f}\ \vec{\ell} \arrz r \in \Rules$.
A term $\ell$ is a \emph{pattern} if (a) $\ell$ is a variable, or
(b) $\ell = \apps{\identifier{c}}{\ell_1}{\ell_m}$ with $\identifier{
c} : \atype_1 \arrtype \dots \arrtype \atype_m \arrtype \asort \in
\Constructors$ for $\asort \in \Sorts$ and all $\ell_i$ patterns.
A \emph{data term} is a pattern without variables,
and the set of all data terms is denoted $\Data$.
A term $\apps{\identifier{f}}{\ell_1}{\ell_n}$ of base type, with
$\identifier{f} \in \Defineds$ and all $\ell_i \in \Data$ data terms
is called a \emph{basic term}.
Note that all non-variable patterns---which includes all
data terms---also have base type.
\end{definition}

We will particularly consider \emph{left-linear constructor
rewriting systems}.

\begin{definition}
A \emph{constructor rewriting system} is an ATRS such that all rules
have the form $\apps{\identifier{f}}{\ell_1}{\ell_k} \arrz r$
with $\identifier{f} \in \Defineds$ and all $\ell_i$ patterns.
It is \emph{left-linear} if all rules are left-linear.
%, and \emph{consistent} if for every pair of rules
%$\apps{\identifier{f}}{\ell_1}{\ell_k} \arrz r_1,\ 
%\apps{\identifier{f}}{\ell_1'}{\ell_n'}$ we have $n = k$.
\end{definition}

Left-linear constructor rewriting systems are very common in the
literature on term rewriting.  The higher-order extension of
\emph{patterns} where the first-order definition merely requires
constructor terms corresponds to the typical restrictions in
functional programming languages, where constructors must be fully
applied.
%Consistency is similarly required in functional languages.
However, unlike functional programming languages, we allow for
overlapping rules, and do not impose an evaluation strategy.

\begin{example}
The ATRSs from Examples~\ref{ex:plusfo} and~\ref{ex:plusho} are
left-linear %consistent
constructor rewriting systems.  In
Example~\ref{ex:plusfo}, $\Constructors = \F_\bits$ and $\Defineds =
\{\symb{succ}\}$.  If a rule $\nul\cons\nil \arrz \nil$ were added
to $\Rules_{\symb{count}}$, it would no longer be a constructor
rewriting system as this would force $\cons$ to be in $\Defineds$,
conflicting with rules $\mathsf{(B)}$ and $\mathsf{(C)}$.
A rule such as $\symb{equal}\ n\ n \arrz \one$ would break
left-linearity.
%In Example~\ref{ex:plusho}, replacing rule
%$\mathsf{(B)}$ by $\symb{check}\ \nul\ F\ n\ m \arrz F\ m$ would
%break consistency (due to rule $\mathsf{(C)}$).
\end{example}

\subsection{Deciding problems using rewriting}\label{subsec:deciderewrite}

Like Turing Machines, an ATRS can decide a set $S \subseteq I^+$
(where $I$ is a finite set of symbols).
Consider ATRSs with a signature $\F = \Constructors_I \cup \Defineds$
where $\Constructors_I = \{\nil : \bits,\ \cons : \symbs \arrtype \bits
\arrtype \bits,\strue : \bool,\sfalse : \bool \} \cup \{ \symb{a} :
\symbs \mid a \in I\}$.
There is an obvious correspondence between elements of $I^+$ and data
terms of sort $\bits$; if $x \in I^+$, we write $\encode{x}$ for the
corresponding data term.

\begin{definition}\label{def:acceptance_ATRS}
An ATRS \emph{accepts} $S \subseteq I^+$ if there is a designated
defined symbol $\symb{decide} : \bits \arrtype \bool$ such that, for
every $x \in I^+$ we have $\symb{decide}\ \encode{x} \arrr{\Rules}
\symb{true}$ if{f} $x \in S$.
The ATRS \emph{decides} $S$ if moreover
$\symb{decide}\ \encode{x} \arrr{\Rules} \sfalse$ if{f} $x \notin S$.
\end{definition}

While Jones considered programs \emph{deciding} decision
problems, in this paper we will consider \emph{acceptance}---a
property reminiscent of the acceptance criterion of non-deterministic
Turing machines---because term rewriting is inherently
non-deterministic unless further constraints (e.g., orthogonality) are
imposed. Thus, an input $x$ is ``rejected'' by a rewriting system if
there is no reduction to $\symb{true}$ from
$\symb{decide}\ \encode{x}$.  As evaluation is non-deterministic,
there may be many distinct reductions starting from
$\symb{decide}\ \encode{x}$.

With an eye on future extensions in \emph{functional}
complexity---where the computational complexity of functions, rather
than sets, is considered---our definitions and lemmas will
more generally consider programs which reduce an arbitrary
\emph{basic term} to a \emph{data term}.
%, rather than merely reductions to $\strue$ as in
%Definition \ref{def:acceptance_ATRS}.
However, our main theorems consider only programs with
main symbol $\symb{decide} : \bits \arrtype \bool$.

\section{Cons-free rewriting}\label{sec:limitations}

As we aim
%Since the purpose of this research is
to find groups of programs which
can handle \emph{restricted} classes of Turing-computable problems, we
will impose certain limitations.  We limit interest to the left-linear
%consistent
constructor TRSs from \secshort\ref{sec:atrss}, but
%---as these are clearly Turing-complete---
impose the additional restriction
that they must be \emph{cons-free}.

\begin{definition}\label{def:consfree}
A rule $\ell \arrz r$ is cons-free if for all $r \suptermeq s$: if
$s$ has the form $\apps{\identifier{c}}{s_1}{s_n}$ with $\identifier{
c} \in \Constructors$, then $s \in \Data$ or $\ell \supterm s$.
A left-linear %consistent
constructor ATRS is cons-free if all its
rules are.
\end{definition}

Definition~\ref{def:consfree} corresponds largely to the definitions of
cons-freeness in~\cite{car:sim:14,jon:01}.  In a cons-free
system, it is not possible to build new non-constant data,
as we will see in \secshort\ref{subsec:properties}.

\begin{example}\label{ex:plushoconsfree}
The ATRSs from Examples~\ref{ex:plusfo} and~\ref{ex:plusho} are not
cons-free; in the first case due to rules $\mathsf{(B)}$ and
$\mathsf{(C)}$, in the second due to rule $\mathsf{(F)}$.  To some
extent, we can repair the second case, however: by counting
\emph{down} rather than \emph{up}.
To be exact, we let $n$ be a \emph{fixed} number, assume
that $\symb{s}^n\ \nul$ is given as input to the ATRS, and represent
a number as a finite bitstring $b_0\dots b_{n-1}$ with the most
significant digit first---in contrast to Example~\ref{ex:plusho},
where we used essentially infinite bitstrings $b_0\dots b_{n-1}
\nul\nul\nul\dots$ with the \emph{least} significant digit first.

We can reuse most of the previous rules, but replace the
(non-cons-free) rule $\mathsf{(L)}$ by:
\[
\begin{array}{crclccrcl}
\mathsf{(L.1)} & \symb{fsucchelp}\ \one\ F\ \symb{o} & \arrz & F & \quad &
\mathsf{(L.2)} & \symb{fsucchelp}\ \one\ F\ (\symb{s}\ n) & \arrz &
  \symb{fsucc}\ F\ n \\
\end{array}
\]
Now a function $F$ represents $b_0\dots b_{n-1}$ if $F$
reduces $\symb{s}^i\ \symb{o}$ to $b_i$ for $0 \leq i < n$;
since we only consider $n$ bits, $F$ may reduce to anything given data
not of this form.  Then $\symb{fsucc}\ F\ (\symb{s}^n\ \symb{o})$
reduces to a function representing the successor of $F$,
%when
%considered as a number with the \emph{most} significant digit first,
modulo $2^n$ (%the bit string
$1\dots 1$ is reduced to $0\dots 0$).
\end{example}

\begin{remark}
The limitation to left-linear constructor systems is standard, but
also \emph{necessary}: if either restriction is dropped, our
limitation to cons-free systems becomes meaningless, and we retain a
Turing-complete language.
This will be discussed in detail in \secshort\ref{sec:llcatrs}.

As the first two restrictions are necessary to give meaning to the
third, we will consider the limitation to left-linear constructor
ATRSs implicit in the notion ``cons-free''.
\end{remark}

\subsection{Properties of Cons-free Term Rewriting}
\label{subsec:properties}

As mentioned, cons-free term rewriting cannot create new
non-constant data terms.  This
means that the set of data terms that might occur during a reduction
starting in some basic term $s$ are exactly the data terms occurring
in $s$, or those occurring in the right-hand side of some rule.
Formally:

\begin{definition}\label{def:B}
Let $(\F,\Rules)$ be a fixed constructor ATRS.  For a given term
$s$, the set $\B_s$ contains all data terms $t$ such that (i) $s
\suptermeq t$, or (ii) $r \suptermeq t$ for some rule $\ell \arrz r
\in \Rules$.
\end{definition}

$\B_s$ is a set of data terms, is closed under subterms and, since we
have assumed $\Rules$ to be fixed, has a linear number of elements in
the size of
$s$.  The property that no new data is generated by reducing $s$ is
formally expressed by the following result:

\begin{definition}[$\B$-safety]
Let $\B \subseteq \Data$ be a set which (i) is closed under taking
subterms, and (ii) contains all data terms occurring as a subterm of
the right-hand side of a rule in $\Rules$.  A term $s$ is
\emph{$\B$-safe} if for all $t$ with $s \suptermeq t$: if $t$ has the
form $\apps{\identifier{c}}{t_1}{t_m}$ with $\identifier{c} \in
\Constructors$, then $t \in \B$.
\end{definition}

\begin{lemma}\label{lem:safetypreserve}
If $s$ is $\B$-safe and $s \arr{\Rules} t$, then $t$ is $\B$-safe.
\end{lemma}

\begin{proof}
By induction on the form of $s$; the result follows trivially by the
induction hypothesis if the reduction does not take place at the head
of $s$, leaving only the base case $s = \apps{\apps{\identifier{f}}{
(\ell_1\gamma)}{(\ell_k\gamma)}}{s_1}{s_n} \arr{\Rules}
\apps{r\gamma}{s_1}{s_n} = t$ for some rule $\apps{\identifier{f}}{
\ell_1}{\ell_k} \arrz r \in \Rules$, substitution $\gamma$ and $n
\geq 0$.
%Write $\ell = \apps{\identifier{f}}{\ell_1}{\ell_k}$.
All subterms $u$ of $t$ are (a) subterms of some $s_i$, (b) subterms
of $r\gamma$ or (c) the term $t$ itself, so
suppose $u = \apps{\identifier{c}}{t_1}{t_m}$ with $\identifier{c}
\in \Constructors$ and consider the three possible situations.

In case (a), $u \in \B$ by $\B$-safety of $s$.

In case (b), either $\gamma(x) \suptermeq u$ for some $x$, or $u =
r'\gamma$ for some $r \suptermeq r' \notin \V$.  In the first case,
$x \in \Var(\ell_i)$ for some $i$ and---since $\ell_i$ is a
pattern---a trivial induction on the form of $\ell_i$ shows that
$\ell\gamma \suptermeq \gamma(x) \suptermeq u$, so again $u \in
\B$ by $\B$-safety of $s = \ell\gamma$.  In the second case, if
$r' = \apps{x}{r_1}{r_{n'}}$ with $x \in \V$ and $n > 0$ then $s
\suptermeq \gamma(x)$ as before, so $\gamma(x) \in \Data$ (because
$\gamma(x)$ must have a constructor as its head), which imposes $n
= 0$; contradiction.
Otherwise $r' = \apps{\identifier{c}}{r_1}{r_n}$, so by definition
of cons-freeness, either $u = r' \in \B$ or $s \supterm \ell_i\gamma
\suptermeq r'\gamma = u$.

In case (c), $n = 0$ because, following the analysis above, $r\gamma
\in \B$.
\end{proof}

Thus, if we start with a basic term $\apps{\identifier{f}}{s_1}{s_n}$,
any data terms occurring in a reduction $\identifier{f}\ \vec{s}
\arrr{\Rules} t$ (directly or as subterms) are in $\B_{\identifier{f}
\ \vec{s}}$.
This insight will be instrumental in \secshort\ref{sec:algorithm}.

\begin{example}\label{ex:majority}
By Lemma~\ref{lem:safetypreserve}, functions in a cons-free ATRS
cannot build recursive data.  Therefore it is often necessary to
``code around'' a problem.  Consider the task of finding the most
common bit in a given bit string.
A typical solution employs a rule like $\symb{majority}\ cs \arrz
\symb{cmp}\ (\symb{count0}\ cs)\ (\symb{count1}\ cs)$.
Now, however, we cannot define $\symb{count}$ functions
which may return arbitrary terms of the form $\symb{s}^i\ \symb{o}$.
Instead we use subterms of the input as a measure of size,
representing a number $i$ by a list of length $i$.
\[
\begin{array}{rclcrcl}
\symb{majority}\ cs & \arrz & \symb{count}\ cs\ cs\ cs \\
\symb{count}\ (\nul\cons xs)\ ys\ (b\cons zs) & \arrz &
  \symb{count}\ xs\ ys\ zs & \quad &
\symb{cmp}\ \nil\ zs & \arrz & \one \\
\symb{count}\ (\one\cons xs)\ (b\cons ys)\ zs & \arrz &
  \symb{count}\ xs\ ys\ zs & &
\symb{cmp}\ (y\cons ys)\ \nil & \arrz & \nul \\
\symb{count}\ \nil\ ys\ zs & \arrz & \symb{cmp}\ ys\ zs & &
\symb{cmp}\ (y\cons ys)\ (z\cons zs) & \arrz & \symb{cmp}\ ys\ zs \\
\end{array}
\]
(The signature extends $\F_\bits$, but is otherwise omitted as types
can easily be derived.)
\end{example}

Through cons-freeness, we obtain another useful property: we do not
have to consider constructors which take functional arguments.

\begin{lemma}\label{lem:niceconstructor}
Given a cons-free ATRS $(\F,\Rules)$ with $\F = \Defineds \cup
\Constructors$, let $Y = \{ \identifier{c} : \atype \in \Constructors
\mid \mathit{order}(\atype) > 1\}$.  Define $\F' := \F \setminus Y$,
and let $\Rules'$ consist of those rules in $\Rules$ not using any
element of $Y$ in either left- or right-hand side.  Then
%\begin{itemize}
%\item
  (a) all data terms and $\B$-safe terms are in $\Terms(\F',\emptyset)$, and
%\item
  (b) if $s$ is a basic term and $s \arrr{\Rules} t$, then
  $t \in \Terms(\F',\emptyset)$ and $s \arrr{\Rules'} t$.
%\end{itemize}
\end{lemma}

\begin{proof}
Since data terms have base type, and the subterms of data
terms are data terms, we have (a).
%; therefore, $\B$-safe terms are in $\Terms(\F',\V)$.
Thus $\B$-safe terms can only
%Since $\B$-safe terms consequently can only
be matched by rules in $\Rules'$,
so Lemma~\ref{lem:safetypreserve} gives (b).
\end{proof}
\vspace{-0.75em}
%Therefore we may safely assume that all elements of $\Constructors$
%are at most first-order.
\subsection{A larger example}\label{subsec:sat}

%None of our examples so far have
%taken advantage of the native non-determinism of term rewriting.
So far, all our examples have been deterministic.
To %demonstrate
show the possibilities, we
consider a first-order cons-free ATRS that solves the Boolean
satisfiability problem (SAT).  This is striking because, in Jones'
language in~\cite{jon:01}, first-order programs cannot
%solve this problem
do this
unless P = NP, even if a non-deterministic $\symb{choose}$
operator is added~\cite{DBLP:conf/amast/Bonfante06}.
The crucial difference is that we, unlike Jones, do not employ a
call-by-value %evaluation
strategy.

Given $n$ boolean variables $x_1,\dots,x_n$ and a boolean formula
$\psi ::= \varphi_1 \wedge \dots \wedge \varphi_m$, the satisfiability
problem considers whether there is an assignment of each $x_i$ to
$\top$ or $\bot$ such that $\psi$ evaluates to $\top$.  Here, each
clause $\varphi_i$ has the form $a_{i,1} \vee \dots \vee a_{i,k_i}$,
where each literal $a_{i,j}$ is either some $x_p$ or $\neg x_p$.  We  
represent this decision problem as a string over $I:=\{\nul,\one,
\mathtt{\#},\mathtt{?}\}$:
the formula $\psi$ is represented by $E::=b_{1,1}\dots b_{1,n}\#
b_{2,1}\dotsb_{2,n}\#\dots\# b_{m,1}\dots b_{m,n}\#$, where for each
$i,j$:
$b_{i,j}$ is $\one$ if $x_j$ is a literal in $\varphi_i$,
$b_{i,j}$ is $\nul$ if $\neg x_j$ is a literal in $\varphi_i$, and
$b_{i,j}$ is $\mathtt{?}$ otherwise.

\begin{example}\label{sat:base}
The satisfiability problem for $(x_1 \vee \neg x_2) \wedge (x_2 \vee
\neg x_3)$ is encoded as $E := 10?\#?10\#$.  Encoding this string as
a data term, we obtain $\encode{E} =
\one\cons\nul\cons\q\cons\h\cons\q\cons\one\cons\nul\cons\h\cons\nil$.
\end{example}

Defining $\Constructors_I$ as done in
\secshort\ref{subsec:deciderewrite} and assuming other declarations
clear from context, we claim that the system in Figure~\ref{fig:sat}
can reduce $\symb{decide}\ \encode{E}$ to
$\symb{true}$ if and only if $\psi$ is satisfiable.

\begin{figure}[th]
\texttt{//} Rules using $\unknown{a},\unknown{b}$ stand for several
rules once: $\unknown{a},\unknown{b}$ range over $\{\symb{0},
\symb{1},\symb{?}\}$ (but not $\symb{\#}$).\phantom{ABC}
\[
%\left.
\begin{array}{rclrcl}
\symb{equal}\ (\h\cons xs)\ (\h\cons ys) & \arrz & \strue \quad &
\symb{equal}\ (\h\cons xs)\ (\unknown{a}\cons ys) & \arrz & \sfalse \\
\symb{equal}\ \nil\ ys & \arrz & \sfalse &
\symb{equal}\ (\unknown{a}\cons xs)\ (\h\cons ys) & \arrz & \sfalse \\
& & & 
\symb{equal}\ (\unknown{a}\cons xs)\ (\unknown{b}\cons ys) & \arrz &
  \symb{equal}\ xs\ ys \\
\symb{either}\ xs\ yss & \arrz & xs &
\symb{skip}\ (\h\cons xs) & \arrz & xs \\
\symb{either}\ xs\ yss & \arrz & yss &
\symb{skip}\ (\unknown{a}\cons xs) & \arrz & \symb{skip}\ xs
\end{array}
%\right\}\ \llbracket\text{for}\ \unknown{a},\unknown{b} \in \{\nul,
%  \one,\symb{?}\}\rrbracket
\]
\vspace{-4pt}
\[
\begin{array}{rcl}
\symb{decide}\ cs & \arrz & \symb{assign}\ cs\ \nil\ \nil\ cs \\
\symb{assign}\ (\h\cons xs)\ yss\ zss\ cs & \arrz &
  \symb{main}\ yss\ zss\ cs \\
%\left.
\symb{assign}\ (\unknown{a}\cons xs)\ yss\ zss\ cs & \arrz &
  \symb{assign}\ xs\ (\symb{either}\ xs\ yss)\ zss\ cs \\
\symb{assign}\ (\unknown{a}\cons xs)\ yss\ zss\ cs & \arrz &
  \symb{assign}\ xs\ yss\ (\symb{either}\ xs\ zss)\ cs \\
%\right\}\ \llbracket\text{for}\ \unknown{a} \in \{\nul,\one,\symb{?}\}
%  \rrbracket
\symb{main}\ yss\ zss\ (\q\cons xs) & \arrz & \symb{main}\ yss\ zss\ xs \\
\symb{main}\ yss\ zss\ (\nul\cons xs) & \arrz &
  \symb{membtest}\ yss\ zss\ xs\ (\symb{equal}\ zss\ xs)\ (\symb{equal}\ yss\ xs) \\
\symb{main}\ yss\ zss\ (\one\cons xs) & \arrz &
  \symb{membtest}\ yss\ zss\ xs\ (\symb{equal}\ yss\ xs)\ (\symb{equal}\ zss\ xs) \\
\symb{main}\ yss\ zss\ (\h\cons xs) & \arrz & \sfalse \\
\symb{main}\ yss\ zss\ \nil & \arrz & \strue \\
\symb{membtest}\ yss\ zss\ xs\ \strue\ b & \arrz &
  \symb{main}\ yss\ zss\ (\symb{skip}\ xs) \\
\symb{membtest}\ yss\ zss\ xs\ b\ \strue & \arrz & \symb{main}\ yss\ zss\ xs \\
\end{array}
\]
\caption{A cons-free first-order ATRS solving the satisfiability problem.}
\label{fig:sat}
\end{figure}

% \medskip
In this system, we follow some of the same ideas as in
Example~\ref{ex:majority}.  In particular, any list of the form
$b_{i+1}\cons\dots\cons b_n\cons\h\dots$ with each $b_j \in \{\nul,\one,
\q\}$ is considered to represent the number $i$ (with $\h\cons\dots$
representing $n$). The rules for $\symb{equal}$ are defined so that
$\symb{equal}\ s\ t$ tests equality of these \emph{numbers}, not the
full lists.
The key idea new to this example is that we use terms not in normal
form to represent a \emph{set} of numbers.
%If we are interested in numbers in $\{1,\dots,n\}$, then
Fixing $n$,
a set $X \subseteq \{1,\dots,n\}$ is
encoded as a pair $(yss,zss)$ of terms such that, for $i \in \{1,\dots,
n\}$:
  $yss \arrr{\Rules} xs$ for a representation $xs$ of $i$
  if and only if $i \in X$,
and
  $zss \arrr{\Rules} xs$ for a representation $xs$ of $i$
  if and only if $i \notin X$.

These pairs ($yss$,$zss$) are constructed using the symbol
$\mathtt{either}$, which is defined by a pair of overlapping rules:
$\mathtt{either}\ s_1\ (\mathtt{either}\ s_2\ (\dots\ 
(\mathtt{either}\ s_{n-1}\ s_n)\dots))$ reduces to each
%of the subterms $s_1,\dots,s_n$.
$s_i$.
We can use such terms as we do---copying and passing
them around without reducing to normal form---because
we do not use %a call-by-value or similar reduction strategy:
call-by-value or similar strategies:
the
ATRS may be evaluated using, e.g., outermost reduction.  While
%it is \emph{possible} to
we \emph{can} use other strategies, any evaluation which
reduces $yss$ or $zss$ too eagerly just ends in an
irreducible, non-data state. %rather than $\strue$ or $\sfalse$.}

Now, an evaluation starting in $\symb{decide}\ \encode{E}$ first
non-deterministically constructs a ``set''$X$---represented
as $(yss,zss)$---containing those boolean
variables assigned $\symb{true}$: $\symb{decide}\ \encode{E}
\arrr{\Rules} \symb{main}\ yss\ zss\ \encode{E}$.
Then, the main function goes through $\encode{E}$, finding for each
clause a literal that is satisfied by the assignment.
Encountering $b_{i,j} \neq \symb{?}$, we determine if $j
\in X$ by comparing both a reduct of $yss$ and of $zss$ to $j$.  If
$yss \arrr{\Rules}$ ``$j$'' then $j \in X$, if $zss
\arrr{\Rules}$ ``$j$'' then $j \notin X$; in either case, we continue
accordingly.  If the evaluation state is incorrect, or if $yss$ or
$zss$ %both non-deterministically reduce
are both reduced to some other term, the
evaluation gets stuck in a non-data normal form.

Note: variable namings are indicative of their use: in an
evaluation starting in $\symb{decide}\ \encode{E}$, the variables $xs$
and $ys$ are always instantiated by data term lists, and $cs$ by
$\encode{E}$; variables $yss$ and $zss$ are instantiated by terms of
type list which do not need to be in normal form.

\begin{example}\label{ex:sat}
To determine satisfiability of $(x_1 \vee \neg x_2) \wedge (x_2 \vee \neg
x_3)$, we reduce $\symb{decide}\ \encode{E}$, where $E = 10?\#?10\#$.
First, we build a valuation.  The %choices made by the
$\symb{assign}$ rules are non-deterministic, but a possible reduction is
$\symb{decide}\ \encode{E} \arrr{\Rules}
\symb{main}\ s\ t\ \encode{E}$, where $s =
\symb{either}\ \encode{0?\#?10\#}\ \nil$ and
$t = \symb{either}\ \encode{\#?10\#}\ (\symb{either}\ 
  \encode{?\#?10\#}\ \nil)$.
%Recall that, since
Since $n = 3$, $\encode{0?\#?10\#}$ represents $1$ while
$\encode{\#?10\#}$ and $\encode{?\#?10\#}$ represent $3$ and $2$
respectively.
Thus, %this corresponds to the valuation
we have
$[x_1:=\top,x_2:=\bot,x_3:= \bot]$.

Then the main loop recurses over the problem. Since
$s$ reduces to a term $\encode{0?\#\dots}$ and $t$ to both
$\encode{\#\dots}$ and $\encode{?\#\dots}$ we have
$\symb{main}\ s\ t\ \encode{E} = \symb{main}\ s\ t\ \encode{10?\#?10\#}
\arrr{\Rules} \symb{main}\ s\ t\ (\symb{skip}\ \encode{10\#?10\#})
\linebreak\arrr{\Rules} \symb{main}\ s\ t\ \encode{?10\#}$:
the first clause is confirmed since %$x_1$ is mapped to $\top$
$x_1:=\top$, so it
%so the clause
is removed and the loop continues with the second clause.
Next, the loop passes over those variables whose assignment does not
contribute to %verifying this clause
the clause, until the clause is confirmed due
to $x_3$:
$\symb{main}\ s\ t\ \encode{?01\#} \arr{\Rules}
\symb{main}\ s\ t\ \encode{01\#} \arrr{\Rules}
\symb{main}\ s\ t\ \encode{1\#} \arrr{\Rules}
\symb{main}\ s\ t \ (\symb{skip}\ \encode{\#}) \arr{\Rules}
\symb{main}\ s\ t\ \nil \arr{\Rules} \strue$.
\end{example}

Due to non-determinism, the term in Example~\ref{ex:sat} could also
have been reduced to $\sfalse$, %simply
by selecting a different
valuation.  This is not problematic: by definition, the ATRS
accepts the set of satisfiable formulas if: $\symb{decide}\ \encode{E}
\arrr{\Rules} \strue$ iff $E$ is a satisfiable formula.
%Thus,
%false negatives or reductions which do not end in a data state are
%allowed.

\section{Simulating $\etime{k}$ Turing machines}\label{sec:counting}

We now show how to simulate Turing Machines by cons-free rewriting.  For this, we use
an approach very similar to that by Jones~\cite{jon:01}.
Fixing a machine $(I,A,S,T)$, we let $\Constructors := \Constructors_A
\cup \{ \symb{s} : \symb{state} \mid s \in S \} \cup \{ \symb{fail} :
\symb{state},\ \symb{L} : \symb{direction},\ \symb{R} :
\symb{direction},\ \symb{action} : \symbs \arrtype \symb{direction}
\arrtype \symb{state} \arrtype \symb{trans} \}$; we denote $\symb{B}$
for the symbol corresponding to $\blank \in A$.
We will introduce defined symbols and rules such that, for any
string $E = c_1\dots c_n \in I^+$:
\begin{itemize}
\item $\symb{decide}\ \encode{E} \arrr{\Rules} \strue$ if{f}
  $(\blank c_1 \dots c_n\blank\blank\dots,0,\symb{start}) \arrzt^*
  (t,i,\symb{accept})$ for some $t,i$;
\item $\symb{decide}\ \encode{E} \arrr{\Rules} \sfalse$ if{f}
  $(\blank c_1 \dots c_n\blank\blank\dots,0,\symb{start}) \arrzt^*
  (t,i,\symb{reject})$ for some $t,i$.
\end{itemize}
While $\symb{decide}\ \encode{E}$ may have other normal forms,
only one normal form will be a data term.

\subsection{Core simulation}

The idea of the simulation is to represent non-negative integers
as terms and let $\symb{tape}\ n\ p$ reduce to the symbol at position
$p$ on the tape at the start of the $n^{\text{th}}$ step, while
$\symb{state}\ n\ p$ returns the state of the machine at time $n$,
\emph{provided the tape reading head is at position $p$}.
If the reading head is not at position $p$ at time $n$, then $\symb{state}\ 
n\ p$ should return $\symb{fail}$ instead; this allows us
to test the position of the reading head. %at any given time.
As the machine is deterministic, we can devise rules to compute these
terms from earlier configurations.

Finding a suitable representation of integers
%and corresponding manipulating functions
is the most intricate part of this simulation,
where we may need higher-order functions and non-deterministic
rules.  Therefore, let us first assume that this can be done.
Then, for a Turing machine which is known to run in time bounded
above by $\lambda n.P(n)$, we define the ATRS in Figure~\ref{fig:TM}
(further elaboration is given as ``comments'' in the ATRS).
As before, the rules are constructed such that, in an evaluation of
$\symb{decide}\ \encode{E}$, the variable $cs$ can always be assumed
to be instantiated by $\encode{E}$.

\begin{figure}[bp]
\[
\begin{array}{rcl}
\multicolumn{3}{l}{
  //\ \text{Determine the transition taken at time}\ \numrep{n}\ 
  \text{given input}\ cs\text{, \emph{provided} the tape}
} \\
\multicolumn{3}{l}{
  //\ \text{reading head is at position}\ \numrep{p}\ \text{at time}
  \ \numrep{n}\text{; if not, reduce to}\ \symb{NA}\ \text{instead.}
} \\
\symb{transition}\ cs\ \numrep{n}\ \numrep{p} & \arrz &
  \symb{transitionhelp}\ (\symb{state}\ cs\ \numrep{n}\ \numrep{p})\ 
                         (\symb{tape}\ cs\ \numrep{n}\ \numrep{p}) \\
\vspace{-4pt}
\symb{transitionhelp}\ \symb{fail}\ x & \arrz & \symb{NA} \\
\symb{transitionhelp}\ \unknown{s}\ \unknown{r} & \arrz &
  \symb{action}\ \unknown{w}\ \unknown{d}\ \unknown{t}\hfill
  \llbracket\text{for all}\ \transition{\unknown{s}}{\unknown{r}}{
  \unknown{w}}{\unknown{d}}{\unknown{t}} \in T\rrbracket \\
\symb{transitionhelp}\ \unknown{s}\ x & \arrz & \symb{end}\ 
  \unknown{s} \hfill\llbracket\text{for}\ \unknown{s} \in
  \{\symb{accept},\symb{reject}\}\rrbracket \\
\end{array}
\]
\[
\begin{array}{rcl}
\multicolumn{3}{l}{
  //\ \text{Determine the state at time}\ \numrep{n}\ \text{given\ 
  input}\ cs\text{, \emph{provided} the tape reading head}
} \\
\multicolumn{3}{l}{
  //\ \text{is at position}\ \numrep{p}\ \text{at time}\ \numrep{n}\ 
  \text{(which happens if it is at position}\ \numrep{p-1},\ 
  \numrep{p}\ \text{or}
} \\
\multicolumn{3}{l}{
  //\ \numrep{p+1}\ \text{at time}\ 
  \numrep{n-1}\ \text{and the right action is taken); if not, reduce
  to}\ \symb{fail}\text{.}
} \\
\symb{state}\ cs\ \numrep{n}\ \numrep{p} & \arrz &
  \symb{ifelse}_{\symb{state}}\ \numrep{n=0}\ 
  (\symb{state0}\ cs\ \numrep{p})\ 
  (\symb{statex}\ cs\ \numrep{n-1}\ \numrep{p}) \\
\symb{state0}\ cs\ \numrep{p} & \arrz &
  \symb{ifelse}_{\symb{state}}\ \numrep{p=0}\ \symb{start}\ 
  \symb{fail} \\
\symb{statex}\ cs\ \numrep{n}\ \numrep{p} & \arrz &
  \symb{statey}\ (\symb{transition}\ cs\ \numrep{n}\ \numrep{p-1})\ 
  (\symb{transition}\ cs\ \numrep{n}\ \numrep{p}) \\
  & & \phantom{\symb{statey}}\ 
  (\symb{transition}\ cs\ \numrep{n}\ \numrep{p+1}) \\
\end{array}
\]
\[
\begin{array}{rclrcl}
\symb{statey}\ (\symb{action}\ x\ \symb{R}\ q)\ a\ e & \arrz & q &
\symb{statey}\ \symb{NA}\ (\symb{action}\ x\ d\ q)\ e & \arrz & 
  \symb{fail} \\
\symb{statey}\ (\symb{action}\ x\ \symb{L}\ q)\ a\ e & \arrz &
  \symb{fail} &
\symb{statey}\ \symb{NA}\ \symb{NA}\ (\symb{action}\ x\ \symb{L}\ q) &
  \arrz & q \\
\symb{statey}\ (\symb{end}\ q)\ a\ e & \arrz & \symb{fail}\ \ \ &
\symb{statey}\ \symb{NA}\ \symb{NA}\ (\symb{action}\ x\ \symb{R}\ q) &
  \arrz & \symb{fail} \\
\symb{statey}\ \symb{NA}\ (\symb{end}\ q)\ e & \arrz & q &
\symb{statey}\ \symb{NA}\ \symb{NA}\ (\symb{end}\ q) & \arrz &
  \symb{fail} \\
\end{array}
\]
\[
\begin{array}{rcl}
\multicolumn{3}{l}{
  //\ \text{Determine the tape symbol at position}\ \numrep{p}\ 
  \text{at time}\ \numrep{n}\ \text{given input}\ cs \text{, which is}
} \\
\multicolumn{3}{l}{
  //\ %\text{the same as}\ 
  \symb{tape}\ cs\ \numrep{n-1}\ 
  \numrep{p}\ \text{unless the transition at time}\ \numrep{n-1}\ 
  \text{occurred at position}\ \numrep{p}\text{.}
} \\
%\multicolumn{3}{l}{
%  \CKchange{//\ \text{position}\ \numrep{p}\ \text{.  Note that}\ 
%  \symb{inputtape}\ cs\ \numrep{p}\ \text{extracts the symbol at
%  position}\ \numrep{p}\ \text{of}
%  }
%} \\
%\multicolumn{3}{l}{
%  \CKchange{//\ \text{the tape at the start of the machine's run
%  (which is}\ \blank c_1\dots c_n\blank\blank\blank\dots\text{).}}
%} \\
\symb{tape}\ cs\ \numrep{n}\ \numrep{p} & \arrz &
  \symb{ifelse}_{\symbs}\ \numrep{n=0}\ 
  (\symb{inputtape}\ cs\ \numrep{p}) \\
  & & \phantom{\symb{ifelse}_{\symbs}\ \numrep{n=0}}\ 
  (\symb{tapex}\ cs\ \numrep{n-1}\ \numrep{p}) \\
\symb{tapex}\ cs\ \numrep{n}\ \numrep{p} & \arrz &
  \symb{tapey}\ cs\ \numrep{n}\ \numrep{p}\ 
  (\symb{transition}\ cs\ \numrep{n}\ \numrep{p}) \\
\symb{tapey}\ cs\ \numrep{n}\ \numrep{p}\ (\symb{action}\ x\ d\ q) &
\arrz & x \\
\symb{tapey}\ cs\ \numrep{n}\ \numrep{p}\ \symb{NA} & \arrz &
  \symb{tape}\ cs\ \numrep{n}\ \numrep{p} \\
\symb{tapey}\ cs\ \numrep{n}\ \numrep{p}\ (\symb{end}\ q) & \arrz &
  \symb{tape}\ cs\ \numrep{n}\ \numrep{p} \\
\symb{inputtape}\ cs\ \numrep{p} & \arrz &
  \symb{ifelse}_{\symbs}\ \numrep{p=0}\ \symb{B}\ 
  (\symb{get}\ cs\ cs\ \numrep{p-1}) \\
\symb{get}\ cs\ \nil\ \numrep{i} & \arrz & \symb{B} \\
\symb{get}\ cs\ (x\cons xs)\ \numrep{i} & \arrz &
  \symb{ifelse}_\symbs\ \numrep{i=0}\ x\ 
  (\symb{get}\ cs\ xs\ \numrep{i-1}) \\
\multicolumn{3}{l}{
  //\ \text{We simulate the TM's outcome by testing whether
  the state at time}\ \numrep{P(|cs|)}\ \text{is}} \\
\multicolumn{3}{l}{
  //\ \,\symb{accept}\ 
  \text{or}\ \symb{reject}\text{, allowing for any reader head
  position in}\ 
  \{\numrep{P(|cs|)},\dots,\numrep{0}\}\text{.}
} \\
\symb{decide}\ cs & \arrz & \symb{findanswer}\ cs\ 
  \symb{fail}\ \numrep{P(|cs|)}\ \numrep{P(|cs|)} \\
\symb{findanswer}\ cs\ \symb{fail}\ \numrep{n}\ \numrep{p}
  & \arrz & \symb{findanswer}\ cs\ (\symb{state}\ cs\ 
  \numrep{n}\ \numrep{p})\ \numrep{n}\ \numrep{p-1} \\
\symb{findanswer}\ cs\ \symb{accept}\ \numrep{n}\ \numrep{p} & \arrz & \strue \\
\symb{teststate}\ cs\ \symb{reject}\ \numrep{n}\ \numrep{p} & \arrz & \sfalse \\
\multicolumn{3}{l}{
  //\ \text{Rules for an if-then-else statement (which is
  not included by default).}
} \\
\symb{ifelse}_{\unknown{\asort}}\ \symb{true}\ y\ z & \arrz & y
\hfill \llbracket\text{for all}\ \unknown{\asort} \in
\{\symb{state},\symbs\}\rrbracket \\
\symb{ifelse}_{\unknown{\asort}}\ \symb{false}\ y\ z & \arrz & z
\hfill \llbracket\text{for all}\ \unknown{\asort} \in
\{\symb{state},\symbs\}\rrbracket \\
\end{array}
\]
\caption{Simulating a deterministic Turing Machine running in
$\lambda x.P(x)$ time.}
\label{fig:TM}
\end{figure}

\subsection{Counting}\label{subsec:counting}

The goal, then, is to represent numbers and define
rules to do four things:
\begin{itemize}
\item calculate $\numrep{P(|cs|)}$ or an overestimation (as the TM
  cannot move from its final state);
\item test whether a ``number'' represents $0$;
\item given $\numrep{n}$, calculate $\numrep{n-1}$,
  \emph{provided $n > 0$}---so it suffices to determine
  $\numrep{\max(n-1,0)}$;
\item given $\numrep{p}$, calculate $\numrep{p+1}$,
  \emph{provided $p+1 \leq P(|cs|)$} as %necessarily
  $\symb{transition}\ cs\ \numrep{n}\ \numrep{p}
  \arr{\Rules} \symb{NA}$ when $n < p$ and $\numrep{n}$ never
  increases---so it suffices to determine
  $\numrep{\min(p+1,P(|cs|))}$.
\end{itemize}
These calculations all occur in the right-hand side of a
rule containing the initial input list $cs$ on the left,
which they can therefore use (for instance to recompute $P(|cs|)$).

Rather than representing a number by a single term, we will use
\emph{tuples} of terms (which are not terms themselves, as
ATRSs do not admit pair types).
%, and adding an}
%explicit pairing constructor would conflict with cons-freeness).
To illus\-trate this, suppose we represent each number $n$ by
a pair $(n_1,n_2)$.
Then the predecessor and successor function must also be split,
e.g.~$\symb{pred}^1\ cs\ n_1\ n_2 \arrr{\Rules} n_1'$ and
$\symb{pred}^2\ cs\ n_1\ n_2 \arrr{\Rules} n_2'$ for $(n_1',n_2')$
some tuple representing $n-1$.  Thus, for instance the
last $\symb{get}$ rule becomes:
\[
\symb{get}\ cs\ (x\cons xs)\ i_1\ i_2 \arrz \symb{ifelse}_{
  \symb{symb}}\ (\symb{zero}\ i_1\ i_2)\ x\ (\symb{get}\ xs\ 
  (\symb{pred}^1\ cs\ i_1\ i_2)\ 
  (\symb{pred}^2\ cs\ i_1\ i_2)
\]

Following Jones~\cite{jon:01}, we use the notion of a \emph{counting
module} which provides an ATRS with a representation of a counting
function and a means of computing.  Counting  modules can be composed,
making it possible to count to greater numbers. Due to the laxity of
term rewriting, our constructions are technically quite different
from those of~\cite{jon:01}.

\begin{definition}[Counting Module]
Write $\F = \Constructors \cup \Defineds$ for the signature in
Figure~\ref{fig:TM}.  For $P$ a function from $\nats$ to $\nats$, a
$P$-counting module of \emph{order $K$} is a tuple
$C_{\pi} ::= (\vec{\atype},\Sigma,R,\A,\numinterpret{\cdot})$%
---where $\pi$ is the name we use to refer to the counting
module---such that:
\begin{itemize}
\item $\vec{\atype}$ is a sequence of types $\atype_1 \otimes \dots
  \otimes \atype_a$ where each $\atype_i$ has order at most $K-1$;
\item $\Sigma$ is a $K^{\text{th}}$-order signature disjoint from
  $\F$, which contains designated symbols $\symb{zero}_{\pi} :
  \bits \arrtype \atype_1 \arrtype \dots \arrtype \atype_a \arrtype
  \bool$ and, for $1 \leq i \leq a$, symbols
  $\symb{pred}_{\pi}^i,\symb{succ}_{\pi}^i :
  \bits \arrtype \atype_1 \arrtype \dots \arrtype \atype_a \arrtype
  \atype_i$ and $\symb{seed}_{\pi}^i : \bits \arrtype \atype_i$
  (and may contain others);
\item $R$ is a set of cons-free (left-linear constructor-)rules
  $\apps{\identifier{f}}{\ell_1}{\ell_k} \arrz r$ with
  $\identifier{f} \in \Sigma$,
  each $\ell_i \in \Terms(\Constructors,\V)$
  and $r \in \Terms(\Constructors \cup \Sigma,\V)$;
\item for every string $cs \subseteq I^+$, %the set
  $\A_{cs} \subseteq
  \{ (s_1,\dots,s_a) \in \Terms(\Constructors \cup \Sigma)^a \mid
  s_j : \atype_j$  for $1 \leq j \leq a\}$;
\item for every string $cs$, $\numinterpret{\cdot}_{cs}$
  is a surjective mapping from $\A_{cs}$ to $\{0,\dots,P(|cs|)-1\}$;
\item the following properties on $\A_{cs}$ and
  $\numinterpret{\cdot}_{cs}$ are satisfied:
  \begin{itemize}
  \item $(\symb{seed}^1_{\pi}\ cs,\dots,\symb{seed}^a_{\pi}\ cs) \in
     \A_{cs}$ and $\numinterpret{(\symb{seed}^1_{\pi}\ cs,\dots,
     \symb{seed}^a_{\pi}\ cs)}_{cs} = P(|cs|)-1$;
  \end{itemize}
  and for all $(s_1,\dots,s_a) \in \A_{cs}$ with $\numinterpret{(s_1,
  \dots,s_a)}_{cs} = m$:
  \begin{itemize}
  \item
    $(\symb{pred}^1_{\pi}\ cs\ \vec{s},\dots,\symb{pred}^a_{\pi}\ cs\ 
    \vec{s})$
    and
    $(\symb{succ}^1_{\pi}\ cs\ \vec{s},\dots,\symb{succ}^a_{\pi}\ cs\ 
    \vec{s})$
    are in $\A_{cs}$;
  \item
    $\numinterpret{(\symb{pred}^1_{\pi}\ cs\ \vec{s},\dots,
    \symb{pred}^a_{\pi}\ cs\ \vec{s})}_{cs} = \max(m-1,0)$;
  \item
    $\numinterpret{(\symb{succ}^1_{\pi}\ cs\ \vec{s},\dots,
    \symb{succ}^a_{\pi}\ cs\ \vec{s})}_{cs} = \min(m+1,P(|cs|)-1)$;
%  \item
%    $\numinterpret{(\symb{inv}^1_{\pi}\ cs\ \vec{s},\dots,
%    \symb{inv}^a_{\pi}\ cs\ \vec{s})}_{cs} = P(|cs|)-1-m$;
  \item $\symb{zero}_{\pi}\ cs\ \vec{s} \arrr{R} \strue$ if{f}
    $m = 0$ and $\symb{zero}_{\pi}\ cs\ \vec{s} \arrr{R} \sfalse$
    if{f} $m > 0$;
  \item if each $s_i \arrr{R} t_i$ and $(t_1,\dots,t_a) \in \A_{cs}$,
    then also $\numinterpret{(t_1,\dots,t_a)}_{cs} = m$.
  \end{itemize}
\end{itemize}
\end{definition}

It is not hard to see how we would use a $P$-counting module in the
ATRS of Figure~\ref{fig:TM}; this results in a $K^{\text{th}}$-order
system for a $K^{\text{th}}$-order module.  Note that number
representations $(s_1,\dots,s_a)$ are not required to be in normal
form: even
if we reduce $\vec{s}$ to some tuple $\vec{t}$, the result of the
$\symb{zero}$ test cannot change from $\strue$ to $\sfalse$ or vice
versa.  As the algorithm relies heavily on these tests, we may
safely assume that terms representing numbers are reduced in a lazy
way---as we did in \secshort\ref{subsec:sat} for the arguments $s$ and
$t$ of $\symb{main}$.

\medskip
To simplify the creation of counting modules, we start by observing
that $\symb{succ}_\pi$ can be expressed in terms of $\symb{seed}_\pi$,
$\symb{pred}_\pi$ and $\symb{zero}_\pi$, as demonstrated in
Figure~\ref{fig:succ} (which also introduces an equality
test, which will turn out to be useful in Lemma~\ref{lem:expQcount}).
In practice, $\symb{succ}_\pi\ cs\ \numrep{n}$
counts down from $\numrep{P(|cs|)-1}$ to some $\numrep{m}$ with
$n = m - 1$.

\begin{figure}[htb]
\vspace{-6pt}
\[
\begin{array}{rcl}
\symb{equal}_\pi\ cs\ n_1 \dots n_a\ m_1 \dots m_a & \arrz &
  \symb{ifelse}_\bool\ (\symb{zero}_\pi\ cs\ \vec{n})\ 
  (\symb{zero}_\pi\ cs\ \vec{m}) \\
  & & \phantom{AB}
  (\ \symb{ifelse}_\bool\ (\symb{zero}_\pi\ cs\ \vec{m})\ 
  \sfalse \\
  & & \phantom{ABCD}
  (\ \symb{equal}_\pi\ cs\ (\symb{pred}^1_\pi\ cs\ \vec{n}) \dots 
  (\symb{pred}^a_\pi\ cs\ \vec{n}) \\
  & & \phantom{ABCD (\ \symb{equal}_\pi\ cs\ }
  (\symb{pred}^1_\pi\ cs\ \vec{m})
  \dots (\symb{pred}^a_\pi\ cs\ \vec{m}) \\
  & & \phantom{ABCD})\ ) \\
\symb{succ}_\pi^i\ cs\ n_1 \dots n_a & \arrz &
  \symb{succ2}_\pi^i\ cs\ n_1 \dots n_a\ (\symb{seed}^1\ cs) \dots
  (\symb{seed}^a\ cs) \\
\symb{succ2}_\pi^i\ cs\ n_1 \dots n_a\ m_1 \dots m_a & \arrz &
  \symb{ifelse}_{\atype_i}\ (\symb{zero}_\pi\ cs\ \vec{m})\ 
  (\symb{seed}^i\ cs)\ (\ \symb{succ3}_\pi^i\ 
  \\ & & \phantom{AB}
  cs\ \vec{n}\ m_i\ (\symb{pred}^1_\pi\ cs\ 
  \vec{m}) \dots (\symb{pred}^a_\pi\ cs\ \vec{m})\ )
  \\
\symb{succ3}_\pi^i\ cs\ n_1 \dots n_a\ m_i\ m_1' \dots m_a' & \arrz &
  \symb{ifelse}_{\atype_i}\ (\symb{equal}_\pi\ cs\ n_1 \dots n_a\ 
  m_1' \dots m_a')\ m_i
  \\ & & \phantom{AB}
  (\ \symb{succ2}_\pi^i\ cs\ n_1 \dots n_a\ m_1' \dots m_a'\ ) \\
\end{array}
\]
\[
\left.
\begin{array}{rcl}
\symb{ifelse}_{\unknown{\btype}}\ \symb{true}\ y\ z & \arrz & y \\
\symb{ifelse}_{\unknown{\btype}}\ \symb{false}\ y\ z & \arrz & z \\
\end{array}
\right\}\ 
\llbracket\text{for}\ \unknown{\btype} \in
\{\bool,\atype_1,\dots,\atype_a\}\rrbracket
\]
\caption{Expressing $\symb{succ}_\pi$ in terms of $\symb{seed}_\pi$,
$\symb{pred}_\pi$ and $\symb{zero}_\pi$.}
\label{fig:succ}
\end{figure}

\begin{remark}\label{rem:nontermination}
Observant readers may notice that the rule for
$\symb{equal}_\pi$ is non-terminating: $\symb{equal}_\pi\ cs\ [0]\ 
[0]$ can be reduced to a term containing $\symb{equal}_\pi\ cs\ [0]\ 
[0]$ as a subterm, as the $\symb{ifelse}$ rules are not prioritised
over other rules.
Following Definition~\ref{def:acceptance_ATRS}, this is unproblematic:
it suffices if there \emph{is} a terminating evaluation from
$\symb{decide}\ \encode{x}$ to $\strue$ if $x \in S$; it is not
necessary for \emph{all} evaluations to terminate.
\end{remark}

\begin{example}\label{ex:countlin}
We design a $(\lambda n.n+1)$-counting module that represents numbers as
(terms reducing to) subterms of the input list $cs$.  Formally, we let
$C_{\mathtt{lin}} := (\symb{list},\Sigma,R,\A,\numinterpret{\cdot})$
where $\A_{cs} = \{ s \in \Terms(\Sigma \cup \Constructors) \mid s :
\symb{list} \wedge s$ has a unique normal form, which is a subterm of
$cs\}$ and $\numinterpret{s}_{cs} =$ the number of $\cons$ operators
in the normal form of $s$. $R$ consists of the rules below along with
the rules in in Figure~\ref{fig:succ}, and $\Sigma$
consists of the defined symbols in $R$.
\[
\begin{array}{rclcrclcrcl}
\symb{seed}^1_{\mathtt{lin}}\ cs & \arrz & cs & \quad &
\symb{pred}^1_{\mathtt{lin}}\ cs\ \nil & \arrz & \nil & &
\symb{zero}^1_{\mathtt{lin}}\ cs\ \nil & \arrz & \strue \\
& & & &
\symb{pred}^1_{\mathtt{lin}}\ cs\ (x\cons xs) & \arrz & xs & \quad &
\symb{zero}^1_{\mathtt{lin}}\ cs\ (x\cons xs) & \arrz & \sfalse \\
\end{array}
\]
\end{example}

The counting module of Example~\ref{ex:countlin} is very simple, but
does not count very high: using it with
Figure~\ref{fig:TM}, we can simulate only machines operating in
$n-1$ steps or fewer.  However, having the linear module as a basis,
we can define \emph{composite} modules to count higher:

\begin{lemma}\label{lem:RQcount}
If there exist a $P$-counting module $C_{\pi}$ and a $Q$-counting module
$C_{\rho}$, both of order at most $K$, then there is a $(\lambda n.P(n)
\cdot Q(n))$-counting module $C_{\pi\cdot\rho}$ of order at most $K$.
\end{lemma}

\begin{proof}
Fixing $cs$ and writing $N := P(|cs|)$ and $M :=
Q(|cs|)$, a number $i$ in $\{0,\dots,N \cdot M - 1\}$ can be seen as
a unique pair $(n,m)$ with $0 \leq n < N$ and $0 \leq m < M$, such
%\pagebreak
that $i = n \cdot M + m$.  Then $\symb{seed}$, $\symb{pred}$ and
$\symb{zero}$ can be expressed using the same
functions on $n$ and $m$.  Write
$C_\pi ::= (\atype_1 \otimes \dots \otimes \atype_a,\Sigma^\pi,
R^\pi,\A^\pi,\numinterpret{\cdot}^\pi)$ and $C_\rho ::= (\btype_1
\otimes \dots \otimes \btype_b,\Sigma^\rho,R^\rho,\A^\rho,
\numinterpret{\cdot}^\rho)$; we assume $\Sigma^\pi$ and $\Sigma^\rho$
are disjoint (wlog by renaming).
Then numbers in $n \in \{0,\dots,
N\}$ are represented in $C_\pi$ by tuples $(u_1,\dots,u_a)$ of length
$a$, and numbers in $m \in \{0,\dots,M\}$ are represented in $C_\rho$
by tuples $(v_1,\dots,v_b)$ of length $b$.  We will represent $n
\cdot M + m$ by $(u_1,\dots,u_a,v_1,\dots,v_b)$.
Formally,
$C_{\pi \cdot \rho} := (\atype_1 \otimes \dots \otimes
\atype_a \otimes \btype_1 \otimes \dots \otimes \btype_b,\Sigma^\pi
\cup \Sigma^\rho \cup \Sigma,R^\pi \cup R^\rho \cup R,\A^{\pi \cdot
\rho},\numinterpret{\cdot}^{\pi \cdot \rho})$, where:
\begin{itemize}
\item $\A^{\pi \cdot \rho} = \{ (u_1,\dots,u_a,v_1,\dots,v_b) \mid
  (u_1,\dots,u_a) \in \A^\pi \wedge (v_1,\dots,v_b) \in \A^\rho \}$,
\item $\numinterpret{(u_1,\dots,u_a,v_1,\dots,v_b)}^{\pi \cdot \rho}_{
  cs} = \numinterpret{(u_1,\dots,u_a)}^\pi_{cs} \cdot Q(|cs|) +
  \numinterpret{(v_1,\dots,v_b)}^\rho_{cs}$,
\item $\Sigma$ consists of the defined symbols in
  $R^\pi \cup R^\rho \cup R$, where $R$ is given by
  Figure~\ref{fig:prod}.
\qedhere
\end{itemize}
\begin{figure}[htb]
\begin{tabular}{ll}
\texttt{//}\hspace{-4pt} & $N \cdot M -1 = (N -1) \cdot M + (M-1)$, which corresponds to
     the pair $(N-1,M-1)$; \phantom{Ab} \\
\texttt{//}\hspace{-4pt}
   & that is, the tuple $(\symb{seed}_\pi^1\ cs,\dots,\symb{seed}_\pi^a\ cs,
    \symb{seed}_\rho^1\ cs,\dots,\symb{seed}_\rho^b\ cs)$.
\end{tabular}
\[
\begin{array}{rcl}
\symb{seed}_{\pi \cdot \rho}^i\ cs & \arrz &
  \symb{seed}_\pi^i\ cs
  \phantom{ABC}\llbracket\text{for}\ 1 \leq i \leq a\rrbracket \\
\symb{seed}_{\pi \cdot \rho}^i\ cs & \arrz & \symb{seed}_{\rho}^{i-a}
  \ cs
  \phantom{AB}\llbracket\text{for}\ a+1 \leq i \leq a+b\rrbracket \\
\end{array}
\]
\begin{tabular}{ll}
\texttt{//}\hspace{-4pt} &
$(n,m)$ represents $0$ if{f} both $n$ and $m$ are $0$.
\phantom{ABCDEFGHIJKLMNOPQRSTUVWXYZ}
\end{tabular}
\[
\begin{array}{rcl}
\symb{zero}_{\pi \cdot \rho}\ cs\ u_1 \dots u_a\ v_1 \dots v_b &
  \arrz & \symb{ifelse}_{\symb{bool}}\ (\symb{zero}_\pi\ 
  cs\ u_1 \dots u_a)\ (\symb{zero}_\rho\ cs\ v_1 \dots v_b) \\
  & & \phantom{\symb{ifelse}_{\symb{bool}}\ (\symb{zero}_\pi\ cs\ u_1 \dots u_a)}\ \sfalse \\
%\symb{and}\ \strue\ x & \arrz & x \\
%\symb{and}\ \sfalse\ y & \arrz & \sfalse \\
\end{array}
\]
\begin{tabular}{ll}
\texttt{//}\hspace{-4pt} &
$(n,m)-1$ results in $(n,m-1)$ if $m > 0$,
  otherwise in $(n-1,M-1)$.
  \phantom{ABCDEFGHIK}
  \\
\end{tabular}
\[
\begin{array}{rcl}
\symb{pred}^i_{\pi \cdot \rho}\ cs\ u_1 \dots u_a\ v_1\ \dots\ v_b &
  \arrz &
  \symb{ptest}_{\pi \cdot \rho}^i\ cs\ (\symb{zero}_\rho\ v_1 \dots
  v_b)\ u_1\ \dots\ u_a\ v_1 \dots v_b \\
  & & \hfill\llbracket\text{for}\ 1 \leq i \leq a+b\rrbracket \\
\symb{ptest}^i_{\pi \cdot \rho}\ cs\ \symb{false}\ \vec{u}\ \vec{v} &
  \arrz & u_i\hfill\llbracket\text{for}\ 1 \leq i \leq a\rrbracket \\
\symb{ptest}^i_{\pi \cdot \rho}\ cs\ \symb{false}\ \vec{u}\ \vec{v} &
  \arrz & \symb{pred}_\rho^{i-a}\ cs\ \vec{v}%v_1 \dots v_b
  \hfill\llbracket\text{for}\ a+1 \leq i \leq a+b\rrbracket \\
\symb{ptest}^i_{\pi \cdot \rho}\ cs\ \symb{true}\ \vec{u}\ \vec{v} &
  \arrz & \symb{pred}_\pi^i\ cs\ \vec{u}%u_1 \dots u_a
  \hfill\llbracket\text{for}\ 1 \leq i \leq a\rrbracket \\
\symb{ptest}^i_{\pi \cdot \rho}\ cs\ \symb{true}\ \vec{u}\ \vec{v} &
  \arrz & \symb{seed}_\rho^{i-a}\ cs\ \vec{v}%v_1 \dots v_b
  \hfill\llbracket\text{for}\ a+1 \leq i \leq a+b\rrbracket \\
\end{array}
\]
\caption{Rules for the product counting module $C_{\pi\cdot\rho}$
(Lemma~\ref{lem:RQcount})}
\label{fig:prod}
\end{figure}
\end{proof}

Lemma~\ref{lem:RQcount} is powerful because it can be used
iteratively.  Starting from the counting module from
Example~\ref{ex:countlin}, we can thus define a first-order
$(\lambda n.(n+1)^a)$-counting module $C_{\mathtt{lin} \cdots
\mathtt{lin}}$ for any $a$.  To reach yet higher numbers, we follow
the ideas from Example~\ref{ex:plusho} and define counting rules
on binary numbers represented as functional terms $F : \vec{\atype}
\arrtype \bool$.

\begin{lemma}\label{lem:expQcount}
If there is a $P$-counting module $C_{\pi}$ of order $K$, then there
is a $(\lambda n.2^{P(n)})$-counting module $C_{p[\pi]}$ of order
$K+1$.
\end{lemma}

\begin{proof}
Write $N := P(|cs|)$ and let $C_\pi = (\atype_1 \otimes
\dots \otimes \atype_a,\Sigma,R,\A,\numinterpret{\cdot
}^\pi)$.
We define the $2^P$-counting module $C_{\symb{p}[\pi]}$ as
$(\atype_1 \arrtype \dots \arrtype \atype_a \arrtype \bool,
\Sigma^{\symb{p}[\pi]},R^{\symb{p}[\pi]},\mathcal{H},
\numinterpret{\cdot}^{\symb{p}[\pi])}$, where:
\begin{itemize}
\item $\mathcal{H}_{cs}$ contains terms $q : \vec{\atype}
  \arrtype \bool$ representing a bitstring $b_0\dots b_{N-1}$ as
  follows: $\apps{q}{s_1}{s_n}$ reduces to $\strue$ if $(s_1,\dots,s_n)$
  represents a number $i$ in $C_\pi$ such
  that $b_i = \one$ and to $\sfalse$ if it represents $i$ with $b_i =
  \nul$.
  Formally, $\mathcal{H}_{cs}$ is the set of all $q \in
  \Terms(\Sigma^{\symb{p}[\pi]} \cup \Constructors,\emptyset)$ of type
  $\atype_1 \arrtype \dots \arrtype \atype_a \arrtype \bool$, where:
  \begin{itemize}
  \item for all $(s_1,\dots,s_a) \in \A_{cs}$: $\apps{q}{s_1}{s_a}$
    reduces to $\strue$ or $\sfalse$, but not both;
  \item for all $(s_1,\dots,s_a),(t_1,\dots,t_a) \in \A_{cs}$: if
    $\numinterpret{(\vec{s})}^\pi_{cs} = \numinterpret{(\vec{t})
    }^\pi_{cs}$---so they represent the same number $i$---%
    then $\apps{q}{s_1}{s_a}$ and $\apps{q}{t_1}{t_a}$
    reduce to the same boolean value.
  \end{itemize}
  For $q \in \mathcal{H}_{cs}$ and $i < N$, we can thus say
  either $q \cdot \numrep{i} \arrr{R^{\symb{p}[\pi]}} \strue$ or $q
  \cdot
  %\pagebreak
  \numrep{i} \arrr{R^{\symb{p}[\pi]}} \sfalse$.
  %, as \emph{any}
  %representation $\numrep{i}$ of $i$ in $C_\pi$ gives the same result.}
\item Let
  $\numinterpret{q}^{\symb{p}[\pi]}_{cs} = \sum_{i=0}^{N-1} \{ 2^{
  N-i-1} \mid \apps{q}{s_1}{s_a} \arrr{R} \strue$ for some
  $(s_1,\dots,s_a)$ with $\numinterpret{(s_1,\dots,s_a)}^\pi_{cs} =
  i \}$.  That is, $q$ represents the number given by the
  bitstring $b_0\dots b_N$
  with $b_N$ the least significant digit (where $b_i = 1$ if and only
  if $q \cdot \numrep{i} \arrr{R^{\symb{p}[\pi]}}
  \strue$).
\item $\Sigma^{\symb{p}[\pi]} = \Sigma \cup \Sigma'$ and
  $R^{\symb{p}[\pi]} = R \cup R'$, where $\Sigma'$ contains all new
  symbols in $R'$, and $R'$ contains the rules below along
  with rules for $\symb{equal}_{\pi}$ and
  $\symb{succ}_{\symb{p}[\pi]}$ following Figure~\ref{fig:succ}.
\end{itemize}
\texttt{//} $\symb{seed}\ cs$ results in a bitstring
that is $1$ at all bits.  We let $\symb{seed}_{\symb{p}[\pi]}\ cs$ be a
normal form: \\
\texttt{//} a term of type $\atype_1 \arrtype \dots \arrtype
\atype_a \arrtype \bool$ which maps all $\numrep{i}$ to $\strue$.
\[
\begin{array}{rcl}
\symb{seed}_{\symb{p}[\pi]}\ cs\ k_1 \dots k_a & \arrz & \strue \\
\end{array}
\]
\texttt{//} A bitstring represents $0$ if all its bits are set to $0$.
To test this, we count down in $C_\pi$ and \\
\texttt{//} evaluate $F\ \numrep{N-1},\ F\ \numrep{N-2},\ \dots,\ F\ 
\numrep{0}$ to see whether any results in $\sfalse$.
\[
\begin{array}{rcl}
\symb{zero}_{\symb{p}[\pi]}\ cs\ F & \arrz & \symb{zero'}_{\symb{p}[
  \pi]}\ cs\ (\symb{seed}^1_\pi\ cs) \dots (\symb{seed}^a_\pi\ cs)\ F \\
\symb{zero'}_{\symb{p}[\pi]}\ cs\ k_1 \dots k_a\ F & \arrz &
  \symb{ifelse}_\bool\ (\apps{F}{k_1}{k_a})\ \sfalse \\
  & & \phantom{AB}(\ \symb{ifelse}_\bool\ (\symb{zero}_\pi\ cs\ k_1 \dots
  k_a)\ \strue \\
  & & \phantom{ABCD}(\ \symb{zero'}_{\symb{p}[\pi]}\ cs\ (\symb{pred
  }^1_\pi\ cs\ \vec{k}) \dots (\symb{pred}^a_\pi\ cs\ \vec{k})\ F\ ) \\
  & & \phantom{AB}) \\
\end{array}
\]
\texttt{//} The predecessor function follows a similar
approach to Examples~\ref{ex:plusho} and~\ref{ex:plushoconsfree}: we
flip $b_i$ \\
\texttt{//} for $i=N-1,N-2,\dots$ until $b_i = \one$
(thus replacing $b_0\dots b_{i-1}\one \nul\dots \nul$
by $b_0\dots b_{i-1}\nul\one\dots\one$).
\begingroup
\addtolength{\jot}{-3.5pt}
\begin{align*}%{rcl}
\symb{pred}_{\symb{p}[\pi]}\ cs\ F & \quad\arrz\quad
  \symb{predtest}_{\symb{p}[\pi]}\ cs\ (\symb{zero}_{\symb{p}[\pi]}\ 
  F)\ cs\ F \\
\symb{predtest}_{\symb{p}[\pi]}\ cs\ \strue\ F & \quad\arrz\quad F \\
\symb{predtest}_{\symb{p}[\pi]}\ cs\ \sfalse\ F & \quad\arrz\quad
  \symb{predhelp}_{\symb{p}[\pi]}\ cs\ F\ (\symb{seed}^1_\pi\ cs)
  \dots (\symb{seed}^a_\pi\ cs) \\
\symb{predhelp}_{\symb{p}[\pi]}\ cs\ F\ \vec{k} & \quad\arrz\quad
  \symb{checkbit}_{\symb{p}[\pi]}\ cs\ (F\ \vec{k})\ %(\apps{F}{k_1}{k_a})\ 
  (\symb{flip}_{\symb{p}[\pi]}\ cs\ F\ \vec{k}%k_1 \dots k_a
  )\ \vec{k} \\
\symb{checkbit}_{\symb{p}[\pi]}\ cs\ \strue\ F\ \vec{k} & \quad\arrz\quad F \\
\symb{checkbit}_{\symb{p}[\pi]}\ cs\ \sfalse\ F\ \vec{k} & \quad\arrz\quad
  \symb{predhelp}_{\symb{p}[\pi]}\ cs\ F\ (\symb{pred}^1_\pi\ cs\ 
  \vec{k}) \dots (\symb{pred}^a_\pi\ cs\ \vec{k}) \\
\symb{flip}_{\symb{p}[\pi]}\ cs\ F\ \vec{k}\ \vec{n} & \quad\arrz\quad
  \symb{ifelse}_\bool\ (\symb{equal}_\pi\ cs\ \vec{k}\ \vec{n})\ 
  (\symb{not}\ (F\ \vec{n}))\ (F\ \vec{n}) \\
\phantom{X_{p[f]}}\symb{not}\ \strue & \quad\arrz\quad \sfalse \\
\phantom{X_{p[f]}}\symb{not}\ \sfalse & \quad\arrz\quad \strue \tag*{\qEd}
\end{align*}
\endgroup
\def\popQED{} % this must be immediately before \end{proof}
\end{proof}

Combining Example~\ref{ex:countlin} with
Lemmas~\ref{lem:RQcount} and~\ref{lem:expQcount}, we can define a
$(\lambda n.\exp_2^{K-1}((n+1)^b))$-counting module
$C_{\symb{p}[\dots[\symb{p}[\symb{lin} \cdots \symb{lin}]]\dots]}$
of type order $K$ for any $K,b \geq 1$.
%With the simulation of Figure~\ref{fig:TM}, we can thus
As the ATRSs of Figure~\ref{fig:TM} and the modules are all
non-overlapping, we thus recover one side of Jones' result: any
problem in $\exptime{K-1}$ is decided using a deterministic
$K^{\text{th}}$-order cons-free ATRS.

\begin{remark}
The construction used here largely follows the one in~\cite{jon:01}.
Differences mostly center around the different formalisms: on the one
hand Jones' language did not support pattern matching or constructors
like $\symb{action}$; on the other, we had to code around the lack of
pairs.  Our notion of a counting module is more complex---restricting
the way tuples of terms may be reduced---to support the
non-deterministic modules we will consider below.
\end{remark}

\subsection{Counting higher}
In ATRSs, we can do better than merely translating Jones'
result. By exploiting non-determinism much like we did in
\secshort\ref{subsec:sat}, we can count up to $2^{n+1}-1$ using only
a first-order ATRS, and obtain the jump in expressivity promised in
the introduction.

\begin{lemma}\label{lem:mainmodule}
There is a first-order ($\lambda n.2^{n+1}$)-counting module.
\end{lemma}

\begin{proof}
Intuitively, we represent a bitstring $b_0\dots b_N$ by a
pair of non-normalized terms $(yss,zss)$, such that $yss \arrr{} [$a
list of length $i]$ iff $b_i = \one$ and $zss \arrr{} [$a list of
length $i]$ iff $b_i = \nul$.  Formally, we let $C_{\symb{e}} :=
(\bits \otimes \bits,\Sigma,R,\A,\numinterpret{\cdot})$,
where:
\begin{itemize}
\item $\A_{cs}$ contains all pairs $(yss,zss)$ such that (a)
  %all data terms $u$ with $yss \arrr{R} u$ or $zss \arrr{R} u$
  all normal forms of $yss$ or $zss$
  are subterms of $cs$,
  and (b) for each $u \unlhd cs$ either
  $yss \arrr{R} u$ or $zss \arrr{R} u$, but not both.
\item Writing $cs = c_N\cons\dots\cons c_1\cons \nil$, we let $cs_i =
  c_i\cons\dots c_1\cons\nil$ for $1 \leq i \leq N$.  Let
  $\numinterpret{(yss,zss)}_{cs} =
  \sum_{i = 0}^N \{ 2^{N-i} \mid yss \arrr{R} cs_i \}$; then
  $\numinterpret{(yss,zss)}_{cs}$ is the number with bit representation
  $b_0\dots b_N$ (most significant digit first) where $b_i
  = 1$ if{f} $yss \arrr{R} cs_i$, if{f} $zss \not\arrr{R} cs_i$.
\item $\Sigma$ consists of the defined symbols introduced in $R$,
  which we construct below.
\end{itemize}
We include the rules from Figure~\ref{fig:succ}, the rules
for $\symb{seed}^1_{\mathtt{lin}},\ \symb{pred}^1_{\mathtt{lin}}$ and
$\symb{zero}^1_{\mathtt{lin}}$ from Example~\ref{ex:countlin}---%which
%we will use
to handle the data lists---and $\symb{ifte}_{\bits}$
defined similar to other $\symb{ifte}$ rules.

As in \secshort\ref{subsec:sat}, we use non-deterministic selection
functions to construct $(yss,zss)$:
\[
\begin{array}{rclcrclcrcl}
\symb{either}\ n\ xss & \arrz & n & \quad & \symb{either}\ n\ xss & \arrz & xss & \quad &
  \bot & \arrz & \bot \\
\end{array}
\]
The symbol $\bot$ will be used for terms which do not reduce to
any data (the $\bot \arrz \bot$ rule serves to force $\bot \in
\Defineds$).
As discussed in Remark~\ref{rem:nontermination},
non-termination by itself is not an issue.
For the remaining functions, we consider bitstring arithmetic.
First, $2^{N+1}-1$ corresponds to the bitstring where each
$b_i = 1$, so $yss$ reduces to all subterms of $cs$:
\[
\begin{array}{rcl}
\symb{seed}_{\symb{e}}^1\ cs & \arrz & \symb{all}\ cs\ 
  (\symb{seed}^1_{\mathtt{lin}}\ cs)\ \bot \\
\symb{seed}_{\symb{e}}^2\ cs & \arrz & \bot \\
\symb{all}\ cs\ n\ xss & \arrz & \symb{ifte}_\bits\ 
  (\symb{zero}^1_{\mathtt{lin}}\ cs\ n)\ (\symb{either}\ n\ xss) \\
  & & \phantom{\symb{ifte}_\bits\ (\symb{zero}^1_{\mathtt{lin}}\ cs\ n}\ \ 
  (\symb{all}\ cs\ (\symb{pred}^1_{\mathtt{lin}}\ cs\ n)\ 
  (\symb{either}\ n\ xss))
\end{array}
\]
%The inverse function is obtained by flipping the sequence's bits:
%\[
%\begin{array}{rclcrcl}
%\symb{inv}_{\symb{e}}^1\ cs\ s\ t & \arrz & t & &
%\symb{inv}_{\symb{e}}^1\ cs\ s\ t & \arrz & s \\
%\end{array}
%\]
(The use of $\symb{seed}^1_{\mathtt{lin}}\ cs$ where simply
$cs$ would have sufficed may seem overly verbose, but is
deliberate because it will make the results of
\secshort\ref{sec:pairing} easier to present.)

In order to define $\symb{zero}_{\symb{e}}$, we must test the value
of all bits in the bitstring.  This is done by forcing an evaluation
from $yss$ or $zss$ to some data term.  This test is constructed in such
a way that both $\strue$ and $\sfalse$ results necessarily reflect the
state of $yss$ and $zss$; any undesirable non-deterministic choices lead
to the evaluation getting stuck.
\[
\begin{array}{rclcrcl}
\symb{eqLen}\ \nil\ \nil & \arrz & \strue & \quad &
\symb{eqLen}\ \nil\ (y\cons ys) & \arrz & \sfalse \\
\symb{eqLen}\ (x\cons xs)\ (y\cons ys) & \arrz &
  \symb{eqLen}\ xs\ ys & &
\symb{eqLen}\ (x\cons xs)\ \nil & \arrz & \sfalse \\
\end{array}
\]
\[
\begin{array}{rclcrcl}
\symb{bitset}\ n\ yss\ zss & \arrz &
  \symb{checkreducts}\ (\symb{eqLen}\ n\ yss)\ 
  (\symb{eqLen}\ n\ zss) \\
\symb{checkreducts}\ \strue\ b & \arrz & \strue \\
\symb{checkreducts}\ b\ \strue & \arrz & \sfalse \\
\end{array}
\]
Then $\symb{zero}_{\symb{e}}\ cs\ yss\ zss$ simply tests whether the bit is
unset for each sublist of $cs$.
\[
\begin{array}{rcl}
\symb{zero}_{\symb{e}}\ cs\ yss\ zss & \arrz & \symb{zo}\ cs\ 
  (\symb{seed}^1_{\mathtt{lin}}\ cs)\ yss\ zss \\
\symb{zo}\ cs\ n\ yss\ zss & \arrz & \symb{ifte}_\bool\ 
  (\symb{bitset}\ n\ yss\ zss)\ 
  \sfalse \\
  & & \phantom{Ab}
  (\ \symb{ifte}_\bool\ (\symb{zero}^1_{\mathtt{lin}}\ cs\ n)\ 
  \strue\ (\symb{zo}\ cs\ (\symb{pred}^1_{\mathtt{lin}}\ cs\ n)\ yss\ zss)\ ) \\
\end{array}
\]
For the predecessor function, we again replace $b_0 \dots
b_{i-1} b 1 0 \dots 0$ by $b_0 \dots b_{i-1} 0 1 \dots 1$.
%In the
%current setting, this requires a bit more work than in
%Lemma~\ref{lem:expQcount}, because we have to
To do so, we
fully rebuild
$yss$ and $zss$.  We first define a helper function $\symb{copy}$ to
copy $b_0 \dots b_{i-1}$:
\[
\begin{array}{rcl}
\symb{copy}\ cs\ n\ yss\ zss\ \sfalse & \arrz &
    \symb{addif}\ (\symb{bitset}\ n\ yss\ zss)\ n \\
  & & \phantom{ABC}
    (\ \symb{copy}\ cs\ (\symb{pred}^1_{\mathtt{lin}}\ cs\ n)\ yss\ zss\ 
    (\symb{zero}_{\mathtt{lin}}\ cs\ n)\ ) \\
\symb{copy}\ cs\ n\ yss\ zss\ \strue & \arrz & \bot \\
\symb{addif}\ \strue\ n\ xss & \arrz & \symb{either}\ n\ xss \\
\symb{addif}\ \sfalse\ n\ xss & \arrz & xss \\
\end{array}
\]
Then, for all $i$,
$\symb{copy}\ cs\ cs_{\max(i-1,0)}\ yss\ zss\ [i=0]$ reduces to
those $cs_j$ with $0 \leq j < i$ where $b_j = 1$, and $\symb{copy}\ cs\ 
cs_{\max(i-1,0)}\ zss\ yss\ [i=0]$ reduces to those with $b_j = 0$.
This works because $yss$ and $zss$ are complements.
To define $\symb{pred}$, we first handle the zero case:
\[
\begin{array}{rcl}
\symb{pred}_{\symb{e}}^1\ cs\ yss\ zss & \arrz &
  \symb{ifte}_\bits\ (\symb{zero}_{\symb{e}}\ cs\ yss\ zss)\ yss\ 
  (\symb{pr}^1\ cs\ (\symb{seed}^1_{\mathtt{lin}}\ cs)\ yss\ zss) \\
\symb{pred}_{\symb{e}}^2\ cs\ yss\ zss & \arrz &
  \symb{ifte}_\bits\ (\symb{zero}_{\symb{e}}\ cs\ yss\ zss)\ zss\ 
  (\symb{pr}^2\ cs\ (\symb{seed}^1_{\mathtt{lin}}\ cs)\ yss\ zss) \end{array}
\]
Then $\symb{pr}\ cs\ cs_N\ yss\ zss$ flips the bits $b_N,b_{N-1},
\dots$ until an index is encountered where $b_i = 1$; this last bit
is flipped, and the remaining bits are copied:
\begingroup
\addtolength{\jot}{-5pt}
\begin{align*}%{rcl}
\symb{pr}^1\ cs\ n\ yss\ zss & \hspace{0.7em}\arrz\hspace{0.7em}
  \symb{ifte}_\bits\ (\symb{bitset}\ n\ yss\ zss) \\
  &  \phantom{\hspace{0.7em} AB\hspace{0.7em}\hspace{0.7em}}
  (\ \symb{copy}\ cs\ (\symb{pred}^1_{\mathtt{lin}}\ cs\ n)\ yss\ zss\ 
  (\symb{zero}_{\mathtt{lin}}\ cs\ n)\ 
  ) \\
  &  \phantom{\hspace{0.7em} AB\hspace{0.7em}\hspace{0.7em}}
  (\ \symb{either}\ n\ (\symb{pr}^1\ cs\ (\symb{pred}^1_{\mathtt{lin}}\ 
  cs\ n)\ yss\ zss)
  \ ) \\
\symb{pr}^2\ cs\ n\ yss\ zss & \hspace{0.7em}\arrz\hspace{0.7em}
  \symb{ifte}_\bits\ (\symb{bitset}\ n\ yss\ zss) \\
  &  \phantom{\hspace{0.7em} AB \hspace{0.7em}\hspace{0.7em}}(\ 
  \symb{either}\ n\ (\symb{copy}\ cs\ (\symb{pred}^1_{\mathtt{lin}}\ 
  cs\ n)\ zss\ yss\ (\symb{zero}_{\mathtt{lin}}\ cs\ n\mathrlap{))
  \ ) }\\
  &  \phantom{\hspace{0.7em} AB\hspace{0.7em}\hspace{0.7em}}(\ 
  \symb{pr}^2\ cs\ (\symb{pred}^1_{\mathtt{lin}}\ cs\ n)\ yss\ zss
  \ ) \tag*{\qEd}
\end{align*}
\endgroup
%Finally, the successor is defined by the rules of
%Figure~\ref{fig:succ}.
\def\popQED{} % this must be immediately before \end{proof}
\end{proof}

Note that, unlike Lemma~\ref{lem:expQcount}, Lemma~\ref{lem:mainmodule}
cannot be used directly to define composite modules: the rules for
$\symb{eqLen}$ rely on the specific choice of the underlying counting
module $C_{\mathtt{lin}}$.  They cannot be replaced by an
$\symb{equals}_{\mathtt{lin}}$ check, because the crucial property is
that---like in \secshort\ref{subsec:sat}---the $\symb{bitset}$
functionality relies on evaluating $yss$ and $zss$ to some normal form.
Nevertheless, even without composing we obtain additional power:

\begin{theorem}\label{thm:simulation}
Any decision problem in $\etime{K}$ is accepted by a
$K^{\text{th}}$-order cons-free ATRS.
\end{theorem}

\begin{proof}
Following the construction in Figure~\ref{fig:TM}, it suffices to
find a $K^{\text{th}}$-order counting module counting up to
$\mathrm{exp}_2^K(a\cdot n)$ where $n$ is the size of the input and
$a$ a fixed positive integer.  Lemma~\ref{lem:mainmodule} gives a
first-order $\lambda n.2^{n+1}$-counting module, and by iteratively
using Lemma~\ref{lem:RQcount} we obtain $\lambda n.(2^{n+1})^a =
\lambda n.2^{a(n+1)}$ for any $a$.  Iteratively applying
Lemma~\ref{lem:expQcount} on the result gives a
$K^{\text{th}}$-order $\lambda n.\mathrm{exp}_2^K(a \cdot (n+1))$-%
counting module.
\end{proof}

\section{Finding normal forms}\label{sec:algorithm}

In the previous section we have seen that every function in
$\etime{K}$ can be implemented by a cons-free $K^{\text{th}}$-order
ATRS.  Towards a characterization result, we must therefore show the
converse: that every function accepted by a cons-free
$K^{\text{th}}$-order ATRS is in $\etime{K}$.

To achieve this goal, we will now give an algorithm running
in $\timecomp{\exp_2^K(a\cdot n)}$ that, on input any basic term in a
fixed ATRS of order $K$, outputs its set of data normal forms.

A key idea is to associate terms of higher-order type to
functions.  For a given set $\B$ of data terms (a shorthand
for a set $\B_s$ following Definition~\ref{def:B}), we let:
\[
\begin{array}{rcl}
\interpret{\asort} & = & \P(\{ s \mid s \in \B\ \wedge
  s : \asort \})\ \ \text{for}\ \asort \in \Sorts\ \ 
  (\text{so $\interpret{\asort}$ is a set of subsets of}\ \B) \\
\interpret{\atype \arrtype \btype} & = & 
  \interpret{\btype}^{\interpret{\atype}}\ \ 
  (\text{so the set of functions from}\ \interpret{\atype}\ 
  \text{to}\ \interpret{\btype}) \\
\end{array}
\]

We will refer to the elements of each $\interpret{\atype}$
as \emph{term representations}.
Intuitively, an element of $\interpret{\asort}$ represents a set of
possible reducts of a term $s : \asort$, while an element of
$\interpret{\atype \arrtype \btype}$ represents the function defined
by a functional term $s : \atype \arrtype \btype$.
Since each $\interpret{\atype}$ is \emph{finite}, we can enumerate
its elements.
In Algorithm~\ref{alg:main} below,
%\CKchange{we will consider ``terms''
%$\apps{\identifier{f}}{A_1}{A_m}$ with each $A_i$ in some
%$\interpret{\atype_i}$}.
we build
functions $\Conf^0,\Conf^1,\dots$, each mapping statements
$\apps{\identifier{f}}{A_1}{A_m} \leadsto t$
to a value in $\{\top,\bot\}$.
Intuitively, $\Conf^i[\apps{\identifier{f}}{A_1}{A_m} \leadsto t]$
denotes whether, in step $i$ in the algorithm, we have confirmed that
$\apps{\identifier{f}}{s_1}{s_m}$ has normal form $t$, where each
$A_j$ represents the corresponding $s_j$.

To achieve this, we will use two helper definitions.
First:
\begin{definition}\label{def:eta}
For a defined symbol $\identifier{f} : \atype_1 \arrtype \dots
\arrtype \atype_m \arrtype \asort \in \Defineds$, rule
$\rho\colon\apps{\identifier{f}}{\ell_1}{\ell_k} \arrz r \in
\Rules$, variables $x_{k+1} : \atype_{k+1},\dots,x_m : \atype_m$ not
occurring in $\rho$ and $A_1 \in \interpret{\atype_1}, \dots,A_m \in
\interpret{\atype_m}$, let the \emph{mapping associated to
$\rho,\ \vec{x}$ and $\identifier{f}\ \vec{A}$} be the function $\eta$
on domain $\{ \ell_j \mid 1 \leq j \leq k \wedge \ell_j \in \V \}
\cup \{ x_{k+1},\dots,x_m \}$ such that $\eta(\ell_j) = A_j$ for
$j \leq k$ with $\ell_j \in \V$, and $\eta(x_j) = A_j$ for $j > k$.
\end{definition}
%(By left-linearity, such a mapping always exists.)
Second, the algorithm employs a function $\NF^i$ for all $i$, mapping
a term $s : \atype$ and a mapping $\eta$ as above to an element
of $%\NF^i(s,\eta) \in
\interpret{\atype}$ (which depends on $\Conf^i$).
Intuitively, if $\delta$ is a substitution such that each $\eta(x)$
represents $\delta(x)$, then $\NF^i(s,\eta)$ represents the term
$s\delta$.

% \begin{figure}[b]
%\medskip\noindent\hrulefill\vspace{-4pt}
\begin{algorithm}\label{alg:main}
\quad

{\bf Input:} A basic term $s = \apps{\identifier{g}}{s_1}{s_M}$.

{\bf Output:} The set of data normal forms of $s$. Note that
this set may be empty.

\smallskip
Set $\B := \B_s$.
For all $\identifier{f} : \atype_1 \arrtype \dots \arrtype \atype_m
\arrtype \asort \in \Defineds$ with $\asort \in \Sorts$, all $A_1 \in
\interpret{\atype_1},\dots,A_m \in \interpret{\atype_m}$, all $t \in
\interpret{\asort}$, let $\Conf^0[\apps{\identifier{f}}{A_1}{A_m}
\leadsto t] := \bot$.
For all such $\identifier{f},\vec{A},t$ and all $i \in \nats$:
\begin{itemize}
\item if $\Conf^i[\identifier{f}\ \vec{A} \leadsto t] = \top$, then
  $\Conf^{i+1}[\identifier{f}\ \vec{A} \leadsto t] := \top$;
\item otherwise, for all $\rho\colon
  \apps{\identifier{f}}{\ell_1}{\ell_k} \arrz r \in
  \Rules$ and fresh variables $x_{k+1} : \atype_{k+1},\dots,x_m :
  \atype_m$,
  all substitutions $\gamma$ on domain $\Var(\identifier{f}\ 
  \vec{\ell}) \setminus \{ \vec{\ell} \}$
  %(so on the variables occurring below constructors)
  such that $\ell_j\gamma \in A_j$ whenever
  $\ell_j \notin \V$,
  %for all $j \leq k$ with $\ell_j \notin \V$ ($A_j \subseteq \B$
  %since $\ell_j$, a non-variable pattern, has base type),
  let $\eta$ be the
  mapping associated to $\rho,\vec{x}$ and $\identifier{f}\ 
  \vec{A}$.
  %function mapping each $\ell_j \in \V$ to $A_j$
  %and each $x_j$ to $A_j$ if $j > k$.
  Test whether $t \in \NF^i((\apps{r}{x_{k+1}}{x_m})\gamma,\eta)$.
  Let $\Conf^{i+1}[\identifier{f}\ \vec{A} \leadsto t]$ be $\top$ if
  there are $\rho,\gamma$ where this test succeeds,
  $\bot$ otherwise.
\end{itemize}
Here, $\NF^i(t,\eta) \in \interpret{\btype}$ is defined recursively for $\B$-safe terms $t : \btype$
and functions $\eta$ mapping all variables $x : \atype$ in $\Var(t)$
to an element of $\interpret{\atype}$, as follows:
\begin{itemize}
\item if $t$ is a data term, then $\NF^i(t,\eta) := \{ t \}$;
\item if $t = \apps{\identifier{f}}{t_1}{t_m}$ with $\identifier{f} :
  \atype_1 \arrtype \dots \arrtype \atype_m \arrtype \asort
  \in \Defineds$ (for $\asort \in \Sorts$), then
  $\NF^i(t,\eta)$ is the set of all $u \in \B$ such that
  $\Conf^i[\apps{\identifier{f}}{\NF^i(t_1,\eta)}{\NF^i(t_m,\eta)}
  \leadsto u] = \top$;
\item if $t = \apps{\identifier{f}}{t_1}{t_n}$ with $\identifier{f} :
  \atype_1 \arrtype \dots \arrtype \atype_m \arrtype \asort
  \in \Defineds$ (for $\asort \in \Sorts$) and $n < m$, then
  $\NF^i(t,\eta) :=$ the function mapping $A_{n+1},\dots,A_m$ to
  the set of all $u \in \B$ such that $\Conf^i[\apps{\apps{
  \identifier{f}}{\NF^i(t_1,\eta)}{\NF^i(t_n,\eta)}}{A_{n+1}}{A_m}
  \leadsto u] = \top$;
\item if $t = \apps{x}{t_1}{t_n}$ with $n \geq 0$ and $x$ a variable,
  then $\NF^i(t,\eta) := \eta(x)(\NF^i(t_1,\eta),\dots,\linebreak
  \NF^i(t_n,\eta))$; so also $\NF^i(t) = \eta(t)$ if $t$ is
  a variable.
\end{itemize}
When $\Conf^{i+1}[\identifier{f}\ \vec{A} \leadsto t] =
\Conf^i[\identifier{f}\ \vec{A} \leadsto t]$ for all statements,
the algorithm ends; we let $I := i+1$ and return $\{ t \in \B \mid
\Conf^I[\apps{\identifier{g}}{\{s_1\}}{\{s_M\}} \leadsto t] = \top\}$.
\end{algorithm}
%\end{figure}

This is well-defined because a non-variable pattern
$\ell_j$ necessarily has base type, which means $A_j$ is a set.
As $\Defineds,\ \B$ and all $\interpret{\atype_i}$ are all finite,
and the number of positions at which $\Conf^i$ is $\top$ increases
in every step, the algorithm always terminates.
The intention is that $\Conf^I$ reflects rewriting for basic terms.
This result is stated formally in
Lemma~\ref{lem:algorithmsoundcomplete}.

\begin{example}
Consider the \emph{majority} ATRS of Example~\ref{ex:majority}, with
starting term $s = \symb{majority}\ (\one\cons\nul\cons\nil)$.  Then
$\B_s = \{ \one,\nul,\one\cons\nul\cons\nil,\nul\cons\nil,\nil\}$.
We have $\interpret{\symbs} = \{ \emptyset,\{\nul\},\{\one\},
\{\nul,\one\} \}$ and $\interpret{\bits}$ is the set containing all
eight subsets of $\{ \one\cons\nul\cons\nil,\ \nul\cons\nil,\ \nil
\}$.
Thus, there are $8 \cdot 2$ statements of the form $\symb{majority}\ 
A \leadsto t$,\ $8^3 \cdot 2$ statements of the form $\symb{count}\ 
A_1\ A_2\ A_3 \leadsto t$ and $8^2 \cdot 2$ of the form $\symb{cmp}\ 
A_1\ A_2 \leadsto t$; in total, 1168 statements are considered in each
step.

We consider one statement in the first step, determining
$\Conf^1[\symb{cmp}\ \{ \nul\cons\nil \}\ \{ \nul\cons\nil,\ \nil \}
\leadsto \nul]$.
There are two viable combinations of a rule and a substitution:
$\symb{cmp}\ (y\cons ys)\ (z\cons zs) \arrz \symb{cmp}\ ys\ zs$ with
substitution $\gamma = [y:=\nul,ys:=\nil,z:=\nul,zs:=\nil]$ and
%\pagebreak
$\symb{cmp}\ (y\cons ys)\ \nil \arrz \nul$ with substitution
$\gamma = [y:=\nul,ys:=\nil]$.
Consider the first.  As there are no functional variables, $\eta$ is
empty and we need to determine whether $\nul \in \NF^1(\symb{cmp}\ 
\nil\ \nil,\emptyset)$.  This fails, because $\Conf^0[\xi] = \bot$ for
all statements $\xi$.  However, the check for the second rule, $\nul
\in \NF^1(\nul,\emptyset)$, succeeds.
Thus, we mark $\Conf^1[\symb{cmp}\ \{ \nul\cons\nil \}\ 
\{ \nul\cons\nil,\ \nil \} \leadsto \nul] = \top$.
\end{example}

Before showing correctness of Algorithm~\ref{alg:main}, we see that it
has the expected complexity.

\begin{lemma}\label{lem:algorithmcomplexity}
If $(\F,\Rules)$ has type order $K$, then Algorithm~\ref{alg:main}
runs in $\timecomp{\exp_2^K(a \cdot n)}$ for some $a$.
\end{lemma}

\begin{proof}
Write $N := |\B|$; $N$ is linear in the size of the only
input, $s$ ($\Rules$ and $\F$ are not considered input).
We claim: if $K,d \in \nats$ are such that $\atype$ has at most
order $K$, and the longest sequence $\atype_1 \arrtype
\dots \arrtype \atype_n \arrtype \asort$ occurring in $\atype$ has
length $n+1 \leq d$, then $\card(\interpret{\atype}) \leq
\mathrm{exp}_2^{K+1}(d^K \cdot N)$.

\begin{quote}
(Proof of claim.)
Proceed by induction on the form of $\atype$.
Observe that $\P(\B)$ has cardinality $2^N$, so for
$\asort \in \Sorts$ also $\card(\interpret{\asort}) \leq 2^N =
\exp_2^1(d^0 \cdot N)$.  For the induction step, write $\atype =
\atype_1 \arrtype \dots \arrtype \atype_n \arrtype \asort$ with $n <
d$ and each $\atype_j$ having order at most $K-1$. We have:
\[\small
\begin{array}{ll}
\card(\interpret{\atype})\!\!\! &=
\card((\cdots (\interpret{\asort}^{\interpret{\atype_n}})^{
\interpret{\atype_{n-1}}}\cdots)^{\interpret{\atype_1}}) 
\\&= \card(\interpret{\asort})^{\card(\interpret{\atype_n})\cdots
  \card(\interpret{\atype_1})} \\
&\leq  2\text{\textasciicircum}(N \cdot \card(\interpret{\atype_n})
  \cdots \card(\interpret{\atype_1}))
 \\&\leq 2\text{\textasciicircum}(N \cdot \mathrm{exp}_2^K(d^{K-1} \cdot N)
  \cdots \mathrm{exp}_2^K(d^{K-1} \cdot N))
  %\ \ \hfill
  %(\text{by IH})
  \\
&= 2\text{\textasciicircum}(N \cdot \mathrm{exp}_2^K(d^{K-1} \cdot N)^n)
\\&\leq 2\text{\textasciicircum}(\mathrm{exp}_2^K(d^{K-1} \cdot N \cdot n +
  N))\ \ \hfill
  (\text{by induction on}\ K \geq 1)  \\
&=  \mathrm{exp}_2^{K+1}(n \cdot d^{K-1} \cdot N + N)
  \\&\leq \mathrm{exp}_2^{K+1}(d \cdot d^{K-1} \cdot N)
  \\&= \exp_2^{K+1}(d^K \cdot N)\ \ \:
  \hfill (n + 1 \leq d) \\
\end{array}
\]
(End of proof of claim.)
\end{quote}
\noindent
Since, in a $K^{\text{th}}$-order ATRS, all arguments types have order
at most $K-1$, we thus find $d$ (depending
solely on $\F$) such that all sets $\interpret{\atype}$ in the
algorithm have cardinality $\leq \mathrm{exp}_2^K(d^{K-1} \cdot
N)$.  Writing $a$ for the maximal arity in $\F$, there are therefore
at most $|\Defineds| \cdot \mathrm{exp}_2^K(d^{K-1} \cdot N)^a
\cdot N \leq |\Defineds| \cdot \mathrm{exp}_2^K((d^{K-1} \cdot a +
1) \cdot N)$ distinct statements $\identifier{f}\ \vec{A} \leadsto t$.

Writing $m := d^{K-1} \cdot a + 1$ and $X := |\Defineds| \cdot
\mathrm{exp}_2^K(m \cdot N)$, we thus find: the algorithm has at most
$I \leq X+2$ steps, and in each step $i$ we consider at most $X$
statements $\varphi$ where $\Conf^i[\varphi] = \bot$.  For every
applicable rule, there are at most $(2^N)^a$ different substitutions
$\gamma$% (for every $\ell_j$ which is not a variable, $\ell_j\gamma$
%must be one of the $\leq 2^N$ elements of $A_j$)
, so we have to test a statement %of the form
$t \in \NF^i((r\ \vec{x})\gamma,\eta)$ at most $X \cdot (X+2)
\cdot |\Rules| \cdot 2^{aN}$ times.  The exact cost of calculating
$\NF^i((r\ \vec{x})\gamma,\eta)$ is implementation-specific, but is
certainly bounded by some polynomial $P(X)$ (which depends on the form
of $r$).
This leaves the total time cost of the algorithm at $\OO(X \cdot (X+1)
\cdot 2^{aN} \cdot P(X)) = P'(\mathrm{exp}_2^K(m \cdot N))$ for
some polynomial $P'$ and constant $m$. As $\etime{K}$ is robust under
taking polynomials, the result follows.
\end{proof}

\subsection{Algorithm correctness}\label{subsec:correct}

The one remaining question is whether our algorithm accurately
simulates rewriting.  This is set out in
Lemma~\ref{lem:algorithmsoundcomplete}.

\begin{lemma}\label{lem:algorithmsoundcomplete}
Let $\identifier{g} : \asort_1 \arrtype \dots \arrtype \asort_M
\arrtype \asort \in \Defineds$ and $s_1 : \asort_1,\dots,s_M :
\asort_M, t : \asort$ be data terms.  Then
$\Conf^I[\apps{\identifier{g}}{\{s_1\}}{\{s_M\}} \leadsto t] = \top$
if and only if $\apps{\identifier{g}}{s_1}{s_M} \arrr{\Rules} t$.
(Here, $I$ is the point at which the algorithm stops
progressing, as defined in the last line of Algorithm~\ref{alg:main}.)
\end{lemma}

A key understanding for Lemma~\ref{lem:algorithmsoundcomplete}
is that algorithm~\ref{alg:main} traces \emph{semi-outermost}
reductions:

\begin{definition}
A reduction $s \arrr{\Rules} t$ is semi-outermost if either $s = 
t$, or it has the form $s = \apps{\identifier{f}}{s_1}{s_n}
\arrr{\Rules} \apps{\apps{\identifier{f}}{(\ell_1\gamma)}{(\ell_k
\gamma)}}{s_{k+1}}{s_m} \arr{\Rules} \apps{(r\gamma)}{s_{k+1}}{s_m}
\arrr{\Rules} t$, the sub-reductions $s_i \arrr{\Rules}
\ell_i\gamma$ and $\apps{(r\gamma)}{s_{k+1}}{s_m} \arrr{\Rules} t$ are
semi-outermost, and $s_j = \ell_j\gamma$ whenever $\ell_j \in \V$.
\end{definition}

\begin{proof}[Proof Idea of Lemma~\ref{lem:algorithmsoundcomplete}]
By postponing reductions at argument positions until needed,
we can safely assume that any reduction in a cons-free
ATRS is semi-outermost.  Then, writing $s \thickapprox A$ to indicate
that $s$ is ``represented'' by $A$% (we will define this formally)
, we
prove by induction:
\begin{itemize}
\item if $s_j \thickapprox A_j$ for $1 \leq j \leq m$, then
  $\Conf^I[\apps{\identifier{f}}{A_1}{A_m} \leadsto t]$ if{f}
  $\apps{\identifier{f}}{s_1}{s_m} \arrr{\Rules} t$;
\item if $\delta$ and $\eta$ have the same domain, and both
  $\delta(x) \thickapprox \eta(x)$ for all $x$ and
  $t_j \thickapprox A_j$ for $1 \leq j \leq n$, then
  $t \in \NF^I(s,\eta)(A_1,\dots,A_n)$ if{f}
  $\apps{(s\delta)}{t_1}{t_n} \arrr{\Rules} t$.
\end{itemize}
%\end{proof}
Lemma~\ref{lem:algorithmsoundcomplete} is then obtained as
an instance of the former statement.
\end{proof}

\medskip
To translate this intuition to a formal proof we must
overcome three difficulties: to translate an arbitrary reduction into
a semi-outermost one, to associate terms to term representations, and
to find an ordering to do induction on (as, in practice, neither
induction on the algorithm nor on reduction lengths works very well
with the definition of $\NF^i$).
The first challenge would be easily handled by an induction on terms
if $\arr{\Rules}$ were terminating, but that is not guaranteed.  To
solve this issue, we will define a terminating relation
corresponding to $\arr{\Rules}$.  This will also be very useful for
the latter two challenges.

\begin{definition}[Labeled system]
Let
$
\Flab := \Constructors \cup \{ \identifier{f}_i : \atype \mid \identifier{f} :
  \atype \in \Defineds \wedge i \in \nats \}
$.
For $s \in \Terms(\F,\V)$ and $i \in \nats$, let $\labl_i(s)$ be $s$
with all instances of any defined symbol $\identifier{f}$ replaced by
$\identifier{f}_i$.  For $t \in \Terms(\Flab,\V)$, let $\|t\|$ be $t$
with all symbols $\identifier{f}_i$ replaced by $\identifier{f}$.
Then, let
\[
\Ruleslab = \{ %\identifier{f} \arrz \identifier{f}_i,
\identifier{f}_{i+1} \arrz \identifier{f}_i \mid \identifier{f} \in
\Defineds \wedge i \in \nats \} \cup \{ \apps{\identifier{f}_{i+1}}{
\ell_1}{\ell_k} \arrz \labl_i(r) \mid \apps{\identifier{f}}{\ell_1}{
\ell_k} \arrz r \in \Rules \wedge i \in \nats \}
\]
\end{definition}

Note that constructor terms are unaffected by $\labl_i$ and $\|\cdot\|$.
The ATRS $(\Flab,\Ruleslab)$ is both non-deterministic and infinite in
its signature and rules, but can be used as a reasoning tool because
data normal forms correspond between the labeled and unlabeled system:

\begin{lemma}\label{lem:labeledequiv}
For all $\identifier{f} : \atype_1 \arrtype \dots \arrtype \atype_m
\arrtype \asort \in \Defineds$ and data terms $s_1,\dots,s_m,t$:
\[
\apps{\identifier{f}}{s_1}{s_m} \arrr{\Rules} t\ 
\text{if and only if}\ 
\apps{\identifier{f}_i}{s_1}{s_m} \arrr{\Ruleslab} t\ 
\text{for some}\ i
\]
\end{lemma}

\begin{proof}
The \emph{if} direction is trivial, as $u \arr{\Ruleslab} v$ clearly
implies that $\|u\| \arr{\Rules} \|v\|$ or $\|u\| = \|v\|$.
For the \emph{only if} direction, note that $u \arr{\Rules} v$
implies $\labl_{i+1}(u) \arrr{\Ruleslab} \labl_i(v)$ for any $i$, by
using the labeled rule $\apps{\identifier{f}_{i+1}}{\ell_1}{\ell_k}
\arrz \labl_i(r)$ if the step $u \arr{\Rules} v$ uses rule
$\apps{\identifier{f}}{\ell_1}{\ell_k} \arrz r$ and using the labeled
rules $\identifier{g}_{i+1} \arrz \identifier{g}_i$ to lower the labels
of all other symbols in $u$.
%The lemma follows by
%induction on the length of the reduction $\apps{\identifier{f}}{s_1}{
%s_n} \arrr{\Rules} t$.
\end{proof}

Despite the label decrease, termination of $\arr{\Ruleslab}$
is non-obvious due to variable copying.  For example, a pair of rules
$\symb{f}_1\ (\symb{c}\ F) \arrz F,\ \symb{g}_2\ x \arrz \symb{f}_1\ x
\ x$ with the constructor $\symb{c} : (\asort \arrtype \asort)
\arrtype \asort$ is non-terminating through the term
$\symb{f}_1\ (\symb{c}\ \symb{g}_2)\ (\symb{c}\ \symb{g}_2)$.  In our
setting, such rules can be assumed not to occur by
Lemma~\ref{lem:niceconstructor}, however.  Thus, we indeed obtain:

\begin{lemma}\label{lem:terminating}
There is no infinite $\arrr{\Ruleslab}$ reduction.
\end{lemma}

\begin{proof}
We use a computability argument reminiscent of the one used for the
\emph{computability path ordering}~\cite{bla:jou:rub:08} (CPO
does not apply directly due to our applicative term structure).
%This may be proved using a modification of the \emph{computability
%path ordering} (CPO)~\cite{bla:jou:rub:08}.  As CPO is not defined
%for applicative rewriting, we adapt the methodology to directly show
%termination of $\arr{\Ruleslab}$.
First, we define \emph{computability}
%for terms $s : \atype$
by induction on
types:
(a) $s : \asort \in \Sorts$ is computable if $s$ is terminating: there
is no infinite $\arrr{\Ruleslab}$-reduction starting in
$s$; (b)
$s : \atype \arrtype \btype$ is computable if $s\ t$ is computable
for all computable $t : \atype$.
Note that
%By induction on types, we easily see that
(I) every computable term is terminating and
(II) if $s$ is computable and $s \arr{\Ruleslab} t$, then $t$ is
computable.
Also, (III), if $\ell\gamma$ is computable for a pattern $\ell$,
then $\gamma(x)$ is computable for all $x \in \Var(\ell)$: if $x$ has
base type then $\gamma(x)$ is a subterm of a terminating term by
(I), otherwise (by Lemma~\ref{lem:niceconstructor}) $\ell = x$ and
$\gamma(x) = \ell\gamma$.

We first observe: every variable, constructor symbol and defined
symbol $\identifier{f}_0$ is computable: let $a : \atype_1 \arrtype
\dots \arrtype \atype_m \arrtype \asort$ be such a symbol;
computability follows if $\apps{a}{s_1}{s_m}$ is terminating for all
computable $s_1 : \atype_1,\dots,s_m : \atype_m$.  We use induction on
$(s_1,\dots.s_m)$ (using the product extension of $\arr{\Ruleslab}$,
which is well-founded on computable terms by (I)) and
conclude with (II) and the induction hypothesis since
$\apps{a}{s_1}{s_m}$ can only be reduced by reducing some $s_i$.

Next we see: every defined symbol $\identifier{f}_i$ is computable, by
induction on $i$.  For $\identifier{f}_0$ we are done; for
$\identifier{f}_{i+1} : \atype_1 \arrtype \dots \arrtype_m \arrtype
\asort$ we must show termination of $\apps{\identifier{f}_{i+1}}{s_1
}{s_m}$ for computable $\vec{s}$.  We are done if every reduct is
terminating.  By induction on $\vec{s}$ by $\arr{\Ruleslab}$ as
before, we are done for reduction steps inside any $s_j$.
Also $\apps{\identifier{f}_i}{s_1}{s_m}$ is computable as $i < i+1$.
This leaves only head reductions $\apps{\identifier{f}_{i+1}}{s_1}{
s_m} \arr{\Ruleslab} \apps{(\labl_i(r)\gamma)}{s_{k+1}}{s_m}$ for
some $\apps{\identifier{f}}{\ell_1}{\ell_k} \arrz r \in \Rules$
with each $s_j = \ell_j\gamma$.
Certainly $\apps{(\labl_i(r)\gamma)}{s_{k+1}}{s_m}$ is terminating if
$\labl_i(r)\gamma$ is computable.  We prove this by a third induction
on $r$, observing that each $\gamma(x)$ is computable by (III):

Write $r = \apps{a}{r_1}{r_n}$ with $x \in \V \cup \F$.  Then
  $\labl_i(r)\gamma = \apps{u}{(\labl_i(r_1)\gamma)}{(\labl_i(r_n)
  \gamma)}$ with $u = \gamma(a)$ or $u \in \Constructors$ or $u =
  \identifier{g}_i$; using the observations above
  and the first induction hypothesis, $u$ is computable in all cases.
  By the third induction hypothesis, also each $\labl_i(r_j)\gamma$ is
  computable, so $\labl_i(r)\gamma$ is a base-type application of
  computable terms. %, so terminating by definition.
\end{proof}

Thus we obtain (a slight variation of) the first step of
the proof intuition:

\begin{lemma}\label{lem:semiouter}
If $s \arrr{\Ruleslab}\!\! t \in \Data$ and $s$ is $\B$-safe,
then $s \arrr{\Ruleslab}\!\! t$ by a semi-outermost reduction.
\end{lemma}

\begin{proof}
By induction on $s$ using $\arr{\Ruleslab} \cup \ \mathord{\supterm}$.  If $s = t$
we are done,
otherwise (by $\B$-safety) $s = \apps{\identifier{f}_i}{s_1}{s_n}$ with
$\identifier{f}_i$ not occurring in $t$.  Thus,
%the reduction cannot solely take place at the argument positions:
a head step must be done:
%we have
$s = \apps{\identifier{f}_i}{s_1}{s_n} \arrr{\Ruleslab}
\apps{\apps{\identifier{f}_i}{(\ell_1\gamma)}{(\ell_k\gamma)}}{
s_{k+1}'}{s_n'} \arr{\Ruleslab} \apps{(r\gamma)}{s_{k+1'}}{s_n'}$ for
some rule $\apps{\identifier{f}_i}{\ell_1}{\ell_k} \in \Rules$,
substitution $\gamma$ and $s_{k+1}',\dots,s_n'$ such that
$s_i \arrr{\Ruleslab} \ell_i\gamma$ for $1 \leq i \leq k$ and $s_i
\arrr{\Ruleslab} s_i'$ for $k < i \leq n$.

Now let $\delta := [x:=\gamma(x) \mid x$ occurs as a strict subterm of
some $\ell_j] \cup [\ell_j:=s_j \mid 1 \leq j \leq k \wedge \ell_j$ is
a variable$]$.  Since all variables occurring in a pattern $\ell_j$ are
subterms of $\ell_j$, clearly
$s \arrr{\Ruleslab} \apps{\apps{\identifier{f}_i}{(\ell_1\delta)}{
(\ell_k\delta)}}{s_{k+1}}{s_n} \arr{\Ruleslab} \apps{(r\delta)}{
s_{k+1}}{s_n} \arrr{\Ruleslab} \apps{\apps{\identifier{f}_i}{
(\ell_1\gamma)}{(\ell_k\gamma)}}{s_{k+1}'}{s_n'} \arrr{\Ruleslab} t$.
Then $s_j = \ell_j\delta$ if $\ell_j$ is a variable, and by
Lemma~\ref{lem:safetypreserve} and the
induction hypothesis ($\supterm$ part for each $s_j$ and
$\arr{\Ruleslab}$ part otherwise), all relevant sub-reductions are
semi-outermost.
%($\supterm$ part), $s_i \arrr{\Ruleslab} \ell_i
%\delta$ by a semi-outermost reduction otherwise.  Also by the
%induction hypothesis ($\arr{\Ruleslab}$ part), $\apps{(r\delta)}{
%s_{k+1}}{s_n} \arrr{\Ruleslab} t$ by a semi-outermost reduction.
\end{proof}

The second difficulty of the proof idea is in the way terms are
associated with term representations.  Within the algorithm, a single
term can have \emph{multiple} representations; for example, a term
$s$ which reduces to $\strue$ and $\sfalse$ is represented both by
$\{\sfalse\}$ and $\{\strue,\sfalse\}$.  This is necessary, because
different normal forms are derived at different times, and may depend
on each other; for example, in an ATRS $\{ \symb{or}\ \strue\ x \arrz
\strue,\ \symb{or}\ \sfalse\ x \arrz x,\ \symb{f} \arrz \sfalse,\ 
\symb{f} \arrz \symb{or}\ \symb{f}\ \strue,\ \symb{g} \arrz \symb{h}
\}$, we need to use that $\NF^1(\symb{f}) = \{\sfalse\}$ to obtain
$\NF^2(\symb{f}) = \{ \strue,\sfalse \}$.  To reflect these levels,
we will continue to use labeled terms:

\begin{definition}
Let $\thickapprox$ be the smallest relation such that
$s \thickapprox A$ if we can write $s = \labl_i(t)\delta$ and
$A = \NF^i(t,\eta)$ for some $i,t,\delta,\eta$ such that $\delta$
and $\eta$ have the same domain and each $\delta(x) \thickapprox
\eta(x)$.
Here, $\NF^i := \NF^I$ if $i > I$.
\end{definition}

The final challenge of the proof idea, the induction, can
be handled in the same way: we will use induction on labeled terms
using $\arrr{\Ruleslab}$.  Thus, we are ready for the formal proof:
 
\begin{proof}[Proof of Lemma~\ref{lem:algorithmsoundcomplete}]
Writing $\Conf^{i} := \Conf^I$ for all $i > I$,
we will see, for all relevant $i \in \nats,\identifier{f} \in
\Defineds,u,\vec{s} \in \Terms(\Flab,\V),t \in \B$,
and term representations $\vec{A},D$:
\begin{enumerate}[label=\bf(\Alph*)]
\item
  if $s_j \thickapprox A_j$ for $1 \leq j \leq m$, then
  $\Conf^i[\apps{\identifier{f}}{A_1}{A_m} \leadsto t]$ if and only
  if $q := \apps{\identifier{f}_i}{s_1}{s_m} \arrr{\Ruleslab} t$
  by a semi-outermost reduction;
\item
  if
  %$u : \atype_1 \arrtype \dots \arrtype \atype_m \arrtype \asort$ and
  $s_j \thickapprox A_j$ for $1 \leq j \leq m$ and $u \thickapprox D$,
  then $t \in D(A_1,\dots,A_m)$ if and only if $q :=
  \apps{u}{s_1}{s_m} \arrr{\Ruleslab} t$ by a semi-outermost
  reduction.
\end{enumerate}
This proves the lemma because, for data terms $s_j$, a trivial
induction on the definition of $\thickapprox$ shows that $s_j
\thickapprox A_j$ if{f} $A_j = \{s_j\}$.  Thus:
$\apps{\identifier{g}}{s_1}{s_M} \arrr{\Rules} t$ if and only if
$\apps{\identifier{g}_i}{s_1}{s_M} \arrr{\Ruleslab} t$ for some $i$
(Lemma~\ref{lem:labeledequiv}),
%\pagebreak
if and only if the same holds with a semi-outermost reduction
(Lemma~\ref{lem:semiouter}),
if and only if $\Conf^i[\apps{\identifier{g}}{\{s_1\}}{\{s_M\}}
\leadsto t] = \top$ for some $i$ (A).
Since $\Conf^i[\xi]$ implies $\Conf^I[\xi]$ for all $i$, we have
the required equivalence.

We prove (A) and (B) together by a mutual induction on
$q$, oriented with $\arr{\Ruleslab}\! \cup\ \rhd$.
%, which is terminating
%because $\arr{\Ruleslab}$ is terminating and monotonic.

\smallskip\noindent
\textbf{(A), only if case}.
Suppose $\Conf^i[\apps{\identifier{f}}{A_1}{A_m} \leadsto t] = \top$,
and each $s_j \thickapprox A_j$.  Then $i > 0$, and if
$\Conf^{i-1}[\apps{\identifier{f}}{A_1}{A_m} \leadsto t] = \top$, then
the induction hypothesis yields
$q \arr{\Ruleslab} \apps{\identifier{f}_{i-1}}{s_1}{s_m} \arrr{
\Ruleslab} t$ by the rule $\identifier{f}_i \arrz
\identifier{f}_{i-1}$, so we are done.

Otherwise, there exist a rule $\apps{\identifier{f}}{\ell_1}{\ell_k}
\arrz r \in \Rules$, variables $x_{k+1},\dots,x_m$ and a substitution
$\gamma$ on domain $\Var(\identifier{f}\ \vec{\ell}) \setminus \{
\vec{\ell} \}$ such that (a) $\ell_j\gamma \in A_j$ for
all non-variable $\ell_j$ and (b) $t \in \NF^{i-1}(
(\apps{r}{x_{k+1}}{x_m})\gamma,\eta)$ where $\eta$ maps
each variable $\ell_j$ to $A_j$, and $x_j$ to $A_j$ for $j > k$.
%---so also an
%$\NF$-substitution as $A_n = \NF^{j_n}(s_n,\eta_n)$.

By part (B) of the induction hypothesis---since
%which we may use because
$q \supterm s_j$---(a) implies that (c) $s_j \arrr{\Ruleslab}
\ell_j\gamma$ by a semi-outermost reduction for all non-variable
$\ell_j$.
Now, if we let $\delta := [\ell_j:=s_j \mid 1 \leq j \leq k \wedge
\ell_j \in \V] \cup [x_j:=s_j \mid k < j \leq m]$ we have $\delta(x)
\thickapprox \eta(x)$ for all $x$. This gives:
\[
\begin{array}{ll}
q = \apps{\identifier{f}_i}{s_1}{s_m}
& \arrr{\Ruleslab} (\apps{\apps{\identifier{f}_i}{\ell_1}{\ell_k}}{
  x_{k+1}}{x_m})\gamma\delta\ \ \ (\text{by (c) and definition of
  $\delta$}) \\
& \arr{\Ruleslab} (\apps{\labl_{i-1}(r)}{x_{k+1}}{x_m})\gamma\delta
  \ \ (\text{by the labeled rule for}\ \apps{\identifier{f}}{\ell_1
  }{\ell_k} \arrz r) \\
& = %\labl_{i-1}(\apps{r}{x_{k+1}}{x_m})\gamma\overline{\zeta} =
  \labl_{i-1}((\apps{r}{x_{k+1}}{x_m})\gamma)\delta
\end{array}
\]
Since at least one step is done
%and $t \in \NF^{i-1}((\apps{r}{x_{k+1 }}{x_m})\gamma,\zeta)$ by (b),
and $\labl_{i-1}(v)\delta \thickapprox \NF^{i-1}(v,\eta)$
for the $\B$-safe term $v = (\apps{r}{x_{k+1}}{x_m})\gamma$, we can
use
%$\arr{\Ruleslab}$ part of the
induction hypothesis (B) on observation (b) to derive that
$q \arrr{\Ruleslab} \labl_{i-1}((\apps{r}{x_{k+1}}{x_m})\gamma)
\delta \arrr{\Ruleslab} t$.
This reduction is semi-outermost.

\smallskip\noindent
\textbf{(A), if case}.  Suppose $q = \apps{\identifier{f}_i}{s_1}{
s_m} \arrr{\Ruleslab} t$ by a semi-outermost reduction.  Since $t$
cannot still contain $\identifier{f}_i$, this is not the empty 
reduction, so either
\[
q = \apps{\identifier{f}_i}{s_1}{s_m} \arr{\Ruleslab}
\apps{\identifier{f}_{i-1}}{s_1}{s_m} \arrr{\Ruleslab} t
\]
in which case induction hypothesis (A) gives $\Conf^{i-1}[
\apps{\identifier{f}}{A_1}{A_m} \leadsto t] = \top$, or
\[
q = \apps{\identifier{f}_i}{s_1}{s_m} \arrr{\Ruleslab}
(\apps{\apps{\identifier{f}_i}{\ell_1}{\ell_k}}{x_{k+1}}{x_m})
\gamma
\arr{\Ruleslab} \labl_{i-1}(\apps{r}{x_{k+1}}{x_m})\gamma
\arrr{\Ruleslab} t
\]
for some rule $\apps{\identifier{f}}{\ell_1}{\ell_k} \arrz r \in
\Rules$, substitution $\gamma$ and fresh variables $x_{k+1},\dots,
x_m$.  Here, $\gamma(x_j) = s_j$ for all $j > k$ and
$\gamma(\ell_j) = s_j$ for those $\ell_j$ which are variables.
%As we are done in the first subcase, we consider only the second.
By induction hypothesis (B), $\ell_j\gamma \in A_j$ whenever $\ell_j$
is not a variable.
Splitting $\gamma := \gamma_1 \uplus \gamma_2$---where $\gamma_1$ has
domain $\{ x \mid x$ occurs in some non-variable $\ell_j\}$ and
$\gamma_2$ has the remainder---and writing $\eta_2 := [\ell_j:=A_j
\mid 1 \leq j \leq k \wedge \ell_j$ is a variable$] \cup [x_j:=A_j
\mid k < j \leq m]$, we have $\gamma_2(x) \thickapprox \eta_2(x)$ for
all $x$ in the shared domain.  Therefore $
\labl_{i-1}(\apps{r}{x_{k+1 }}{x_m})\gamma =
(\labl_{i-1}(\apps{r}{x_{k+1}}{x_m})\gamma_1)\gamma_2
\thickapprox \NF^{i-1}((\apps{r}{x_{k+1}}{x_m})\gamma_1,\eta_2)$, and
we obtain $t \in \NF^{i-1}((\apps{r}{x_{k+1}}{x_m})\gamma_1,\eta_2)$
by IH (B).

Thus, in either case, $\Conf^i[\apps{\identifier{f}}{A_1}{
A_m} \leadsto t] = \top$ follows immediately.

\smallskip\noindent
\textbf{(B), both cases}.
We prove (B) by an additional induction
on the definition of $u \approx D$.  Observe that $u \thickapprox D$
implies that $u = \labl_i(v)\delta$ and $D = \NF^i(v,\eta)$ for some
$v,i,\delta,\eta$ such that each $\delta(x) \thickapprox \eta(x)$.
Consider the form of the $\B$-safe term $v$.
\begin{itemize}
\item If $v \in \Data$, then
  $m = 0$ and
  $t \in D = \NF^i(v,\eta)$ if{f} $t = v = \labl_i(v) = u$.
\item If $v = \apps{\identifier{f}}{v_1}{v_n}$ with
  $\identifier{f} \in \Defineds$, then
  denote $C_j := \NF^i(v_j,\eta)$ for $1 \leq j \leq n$;
  we have
  %whether or not $u$ has base type,
  $t \in D(A_1,\dots,A_m)$ if{f}
  $\Conf^i[\apps{\apps{\identifier{f}}{C_1}{C_n}}{A_1}{A_m} \leadsto t]
  = \top$.
  By case (A), this holds if{f} $q =
  \apps{(\labl_i(v)\delta)}{s_1}{s_m} =
  \apps{\apps{\identifier{f}_i}{(\labl_i(v_1)\delta)}{(
  \labl_i(v_n)\delta)}}{s_1}{s_m} \arrr{\Ruleslab} t$.
\item If $v = \apps{x}{v_1}{v_n}$ with $x \in \V$, then
  denote $C_j := \NF^i(v_j,\eta)$ for $1 \leq j \leq n$;
  then clearly $\labl_i(v_j)\delta \thickapprox C_j$.  We observe
  that, on the one hand,
  \[
  \begin{array}{ccccc}
  \phantom{ABC}D(A_1,\dots,A_m) & = & \NF^i(v,\eta)(A_1,\dots,A_m) &
  = & (\eta(x)(C_1,\dots,C_n))(A_1,\dots,A_m) \\
  & & & = & \eta(x)(C_1,\dots,C_n,A_1,\dots,A_m)\phantom{W} \\
  \end{array}
  \]
  And on the other hand,
  \[
  q = \apps{(\labl_i(v)\delta)}{s_1}{s_m}
  = \apps{\apps{\delta(x)}{(\labl_i(v_1)\delta)}{
    (\labl_i(v_n)\delta)}}{s_1}{s_m}
  \]
  As $\delta(x) \thickapprox \eta(x)$ is used in the derivation of
  $u \thickapprox D$, the second induction hypothesis gives the
  desired equivalence.
\qedhere
\end{itemize}
\end{proof}

And from Lemmas~\ref{lem:algorithmcomplexity}
and~\ref{lem:algorithmsoundcomplete} together we obtain:

\begin{theorem}\label{thm:algorithm}
Any decision problem accepted by a cons-free $K^{\text{th}}$-order
ATRS is in $\etime{K}$.
\end{theorem}

\begin{proof}
By Lemma~\ref{lem:algorithmsoundcomplete}, decision problems
accepted by a cons-free $K^{\text{th}}$-order ATRS are decided by
Algorithm~\ref{alg:main}; by Lemma~\ref{lem:algorithmcomplexity},
this algorithm operates within $\bigcup_{a \in \nats}
\timecomp{\textrm{exp}_2^{K}(an)}$.
\end{proof}

\subsection{Characterization result}
Combining Theorems~\ref{thm:simulation}
and~\ref{thm:algorithm} we thus find:

\begin{corollary}\label{cor:main}
A decision problem $X$ is in $\etime{K}$ if{f} there is a
$K^{\text{th}}$-order cons-free ATRS which accepts $X$:
the class of cons-free ATRSs with order $K$ \emph{characterizes}
$\etime{K}$.
\end{corollary}

\begin{remark}
There are many similarities between the algorithm and
correctness proof presented here and those in Jones' work, most
pertinently the use of \emph{memoization}.  We have chosen to use a
methodology which suits better with the semantics of term rewriting
than the derivation trees of~\cite{jon:01}, for example by enumerating
all possible reductions beforehand rather than using caching, but this
makes little practical difference.  We have also had to make
several changes for the non-determinism and different evaluation
strategy.  For example the step to semi-outermost reductions is unique
to this setting, and the term representations are different than
they must be in the deterministic (or call-by-value) cases.
\end{remark}

\section{Pairing}\label{sec:pairing}

Unlike our applicative term rewriting systems, Jones' minimal language
in~\cite{jon:01} includes \emph{pairing}.  While not standard in term
rewriting, some styles of higher-order rewriting also admit pairs.
We consider whether this feature affects expressivity of the
considered systems.

\begin{definition}\label{def:aptrs}
An \emph{Applicative Pairing Term Rewriting System} (APTRS) is
defined following the definitions for ATRSs in
\secshort\ref{sec:atrss}, with the following changes:
\begin{itemize}
\item In Definition~\ref{def:types} (simple types):
  if $\atype,\btype$ are types of order $n,m$, then also $\atype
  \times \btype$ is a type of order $\max(n,m)$;
  the pairing constructor $\times$ is considered right-associative.
\item In Definition~\ref{def:terms} (terms):
  terms are expressions typable by clauses (a), (b), (c), where (c)
  is: $(s,t) : \atype \times \btype$ if $s : \atype$ and $t :
  \btype$.  Pairing is right-associative, so $(s,t,u) =
  (s,(t,u))$.
%\item In Definition~\ref{def:rules}, rules must have the form
%  $\apps{\identifier{f}}{\ell_1}{\ell_k} \arrz r$ (there are no rules
%  $(s,t) \arrz u$).
\item In Definition~\ref{def:data} (patterns, data and basic terms):
  a term $\ell$ is a \emph{pattern} if (a), (b) or (c) holds,
  where (c) is $\ell = (\ell_1,\ell_2)$ with $\ell_1$ and $\ell_2$
  both patterns.
\end{itemize}
The last item is used to define \emph{constructor APTRSs} as before.
\end{definition}

Cons-freeness for left-linear constructor APTRSs is unaltered from
Definition~\ref{def:consfree}; however, pairing is not a
constructor, so may occur freely in both sides of rules.
Lemmas~\ref{lem:safetypreserve} and~\ref{lem:niceconstructor}
go through unmodified, but constructors \emph{can} have a product type
of order $0$ as argument type.
%Moreover, $\B$-safety of a term $s$ does not mean that all \emph{data}
%subterms of $s$ are in $\B$.
%
%\begin{lemma}
%Let $\Sigma := \{ \tau \mid \identifier{f} : \sigma \in \F \wedge
%\tau$ is a subtype of $\sigma \}$, and let $n$ be the length of the
%largest type of order $0$ occurring in $\Sigma$ (where
%$\mathit{length}(\atype) =$ the number of sorts in $\atype$).
%Then for any ground $\B$-safe term $s$ that is not a pair: if
%$s \suptermeq t$ for some data term $t$, then $t \in \B^n$.
%\end{lemma}
%
%\begin{proof}
%$t$ is a data term, so---by a simple induction---must have a type
%$\sigma$ of order $0$; since $t = s$ (not a pair) or $t$ is an
%immediate subterm of a term $\apps{\identifier{f}}{t_1}{t_n}$, the
%length of $\sigma$ is at most $n$.  We show by induction on $i$:
%if $\atype$ has order $0$ and length $i$ and $t : \atype \in \Data$
%is $\B$-safe, then $t \in \B^i$: (a) if GNAH I'd have to define that.
%\end{proof}
%
%Any $\B$-safe data term with a type $\asort_1 \times \dots \times
%\asort_n$ must be in $\B^n$.

In a deterministic setting, pairing makes no difference: a function
$\identifier{f} : (\atype \times \btype) \arrtype \ctype$
%taking an
%argument of product type may
can
be replaced by a function $\identifier{f}
: \atype \arrtype \btype \arrtype \ctype$ with two arguments, and a
function $\identifier{f} : \ctype \arrtype (\atype \times \btype)$
%returning a product may be replaced by two
by \emph{two} functions $\identifier{f}^1
: \ctype \arrtype \atype$ and $\identifier{f}^2 : \ctype \arrtype
\btype$.  We exploited this when defining counting modules
(in~\cite{jon:01}, a number is represented by a single term, which may
have product type).
However, when allowing non-deterministic choice, pairing does increase
expressivity---alarmingly so.

\newcommand{\pairpi}{{(\pi\pi)}}
\begin{lemma}\label{lem:countpairs}
Suppose counting modules are defined over APTRSs.
If there is a first-order $P$-counting module $C_\pi = (\atype \otimes
\btype,\Sigma^\pi,R^\pi,A^\pi,\numinterpret{\cdot}^\pi)$, then there
is a first-order $(\lambda n.2^{P(n)-1})$-counting module $C_\pairpi =
((\atype \times \btype) \otimes (\atype \times \btype),
\Sigma^\pairpi,R^\pairpi,A^\pairpi,\numinterpret{\cdot}^\pairpi)$.
\end{lemma}

\begin{proof}
By using pairing, the ideas of
Lemma~\ref{lem:mainmodule} can be used to create a \emph{composite}
module.  We will use almost the same rules, but replace the
underlying module $C_{\mathtt{lin}}$ by $C_\pi$.
We say $s \mapsto i$ if there is $(t,u) \in \A^\pi_{cs}$ such that $s
\arrr{R} (t,u)$ and $\numinterpret{(t,u)}_{cs}^\pi = i$.
A bitstring $b_0\dots b_N$ is represented by a pair $(yss,zss)$ such
that $yss \mapsto i$ if{f} $b_i = \one$ and $zss \mapsto i$ if{f}
$b_i = \nul$.
\begin{itemize}
\item $\A^\pairpi_{cs}$ contains all pairs $(yss,zss)$ where
  \begin{itemize}
  \item for all $0 \leq i < P(n)$: either $yss \mapsto i$ or $zss \mapsto
    i$, but not both;
  \item if $yss \arrr{R} (u,v)$ then there is $(u',v') \in \A^\pi_{cs}$
    such that $yss \arrr{R} (u',v') \arrr{R} (u,v)$ \\
    (thus, any pair which $yss$ reduces to is a number in $C_\pi$, or
    a reduct thereof);
  \item if $zss \arrr{R} (u,v)$ then there is $(u',v') \in \A^\pi_{cs}$
    such that $zss \arrr{R} (u',v') \arrr{R} (u,v)$.
  \end{itemize}
\item $\numinterpret{(s,t)}^\pairpi_{cs} = \sum_{i=0}^{P(|cs|)-1} \{
  2^{|cs|-i} \mid s \mapsto i \}$.  So $\numinterpret{(s,t)}^\pairpi_{
  cs}$ is the number with bitstring $b_0\dots b_{P(|cs|)-1}$ where
  $b_i = 1$ if{f}
  $s \mapsto i$, if{f} $t \not\mapsto i$ (with $b_0$ the most
  significant digit).
\item $\Sigma^\pairpi$ consists of the defined symbols introduced in
  $R^\pairpi$, which we construct below.
\end{itemize}
The rules for the module closely follow those in
Lemma~\ref{lem:mainmodule}, except that:
\begin{itemize}
\item calls to $\symb{seed}^1_{\mathtt{lin}}$,
  $\symb{zero}_{\mathtt{lin}}$ and $\symb{pred}^1_{\mathtt{lin}}$ are
  replaced by $\symb{seed}_\pi$, $\symb{zero}_\pi$ and
  $\symb{pred}_\pi$ respectively, where these symbols are supported
  by rules such as $\symb{zero}_\pi\ cs\ (s,t) \arrz \symb{zero}_\pi\ 
  cs\ s\ t$ and $\symb{pred}_\pi\ cs\ (s,t) \arrz (\symb{pred}^1_\pi\ 
  cs\ s\ t, \symb{pred}^2_\pi\ cs\ s\ t)$;
\item calls $\symb{eqLen}\ n\ q$ are replaced by $\symb{eqBase}\ cs\ 
  n\ q$, and the rules for $\symb{eqLen}$ replaced by
  $
  \symb{eqBase}\ cs\ (n_1,n_2)\ (m_1,m_2) \arrz \symb{equal}_\pi\ cs\ 
  n_1\ n_2\ m_1\ m_2$.
  Just like a call to $\symb{eqLen}\ n\ q$ forces a reduction from $q$
  to a data term, a call to $\symb{eqBase}\ cs\ n\ q$ forces $q$ to be
  reduced to a pair---but not necessarily to normal form.
\end{itemize}
With these rules, indeed $\symb{seed}^1_\pairpi\ cs$ is in
$\A^\pairpi_{cs}$, as is $\symb{pred}^1_\pairpi\ cs\ n$ if $n \in \A^\pi_{cs}$.
Moreover, we can check that the requirements on reduction are
satisfied.
\end{proof}

Thus, by starting with $C_{\mathtt{e}}$ and repeatedly using
Lemma~\ref{lem:countpairs}, we can reach arbitrarily high
exponential bounds (since %$2^n-1 \geq n$ and therefore
$2^{2^n-1}
\geq 2^n$).  Following the reasoning of \secshort\ref{sec:counting},
%every problem in any $\etime{K}$---so any problem in
%$\elementary$---can therefore be accepted by a first-order cons-free
%APTRS.
we thus have:

\begin{corollary}\label{cor:atleastelem}
Every set in $\elementary$ is accepted by a cons-free
first-order APTRS.
\end{corollary}

The key reason for this explosion in expressivity is that, by matching
on a pattern $(x,y)$, a rule forces a \emph{partial evaluation}.
Recall that, in a cons-free ATRS (without pairing), we can limit
interest to \emph{semi-outermost} reductions, where sub-reductions
$\apps{\identifier{f}}{s_1}{s_n} \arrr{\Rules} \apps{\identifier{f}}{
u_1}{u_n} = \ell\gamma \arr{\Rules} r\gamma$ have
$s_i = u_i$ or $u_i \in \Data$ for all $i$: we can postpone an
evaluation at an argument position if it is
not to a data term.  By allowing a wider range of terms than just the
elements of $\B$ to carry testable information, % that can be directly tested,
expressivity increases accordingly.
%Note that, in the proof of Lemma~\ref{lem:algorithmsoundcomplete},
%we relied heavily on this ability to postpone reduction steps in the
%if case for (A).

We strongly conjecture that it is not possible to accept sets \emph{not}
in $\elementary$, however.
A proof might use a variation of Algorithm~\ref{alg:main}, where
$\interpret{\atype \times \btype} = \{ (A,B) \mid A \in
\interpret{\atype} \wedge B \in \interpret{\btype} \}$: the size of
this set is exponential in the sizes of $\interpret{\atype}$ and
$\interpret{\btype}$, leading to a limit of the form $\exp_2^{a \cdot
n^b}$ depending on the types used.
However, we do not have the space to prove this properly, and the
result does not seem interesting enough to warrant the effort.

\medskip
Yet, product types \emph{are} potentially useful.
We can retain them while suitably constraining expressivity, by
imposing a new restriction.

\begin{definition}
An APTRS is \emph{product-cons-free} if it is cons-free and for all
rules $\apps{\identifier{f}}{\ell_1}{\ell_k} \arrz r$ and subterms
$r \suptermeq (r_1,r_2)$: each $r_i$ has a form (a) $(s,t)$, (b)
$\apps{\identifier{c}}{s_1}{s_n}$ with $\identifier{c} \in
\Constructors$, or (c) $x \in \V$ such that $x \neq \ell_j$ for any
$j$ (so $x$ occurs below a constructor or pair on the left).
\end{definition}

In a product-cons-free APTRS, any pair which is created is
necessarily a data term.  Lemma~\ref{lem:countpairs} does not go
through in a product-cons-free APTRS (due to the rules for
$\symb{pred}_\pi$ and $\symb{seed}_\pi$), but we \emph{do} obtain a
milder increase in expressivity: from $\etime{K}$ to
$\exptime{K}$.

\newcommand{\polmain}{{\symb{exp}\langle a,b \rangle}}
\begin{lemma}\label{lem:countaptrsgood}
Suppose counting modules are defined over product-cons-free APTRSs.
Then for all $a \geq 0, b > 0$, there is a first-order
$(\lambda n.2^{a \cdot (n+1)^b})$-counting module $C_\polmain$.
\end{lemma}

\begin{proof}
As in $C_{\mathtt{e}}$ from Lemma~\ref{lem:mainmodule}, we
will represent a number with bitstring $b_0\dots b_N$ by two terms
$yss$ and $zss$, such that $yss \mapsto i$ if{f} $b_i = \one$ and $zss
\mapsto i$ if{f} $b_i = \nul$.  However, where in
Lemma~\ref{lem:mainmodule} we say $s \mapsto i$ if $s$ reduces to a
data term list of length $i$, here we say $s \mapsto i$ if $s$ reduces
to a data term of type $\bits^{b+1}$ which represents $i$ as in
Lemma~\ref{lem:RQcount}.

Formally:
Write $|xs|$ for the length of a data term list $xs$, so the number of
$\cons$ symbols occurring in it.
Let $\mathit{Base}$ be the set of all data terms $(u_0,\dots,u_b) :
\bits^{b+1}$ such that (a) $|u_0| < a$ and (b) for $0 < i \leq b$:
$|u_i| \leq |cs|$.  We say that $(u_0,\dots,u_b) \in \mathit{Base}$
\emph{base-represents} $k \in \nats$ if $k = \sum_{i = 0}^b |u_i|
\cdot (|cs|+1)^{b-i}$.  (This follows the same idea as
Lemma~\ref{lem:RQcount}.)
For a term $s$, we say $s \mapsto k$ if $s$ reduces by $\arr{R}$ to
an element of $\mathit{Base}$ which base-represents $k$.

Now let $C_\polmain := (\bits^{b+1},\Sigma^\polmain,R^\polmain,
\A^\polmain,\numinterpret{\cdot}^\polmain)$,
where:
\begin{itemize}
\item $\A^\polmain_{cs}$ contains all $(yss,zss)$ such that (a)
  all normal forms of $yss$ or $zss$ are in $\mathit{Base}$,
%  $u
%  \in \mathit{Base}$ for all $u \in \Data$ with $yss \arrr{R} u$ or $zss
%  \arrr{R} u$,
and (b) for all $0 \leq i < a \cdot (|cs|+1)^b$: either
  $yss \mapsto i$ or $zss \mapsto i$, but not both.
\item $\numinterpret{(s,t)}^\polmain_{cs} = \sum_{i=0}^{N} \{
  2^{N-i} \mid s \mapsto i \}$, where $N = a \cdot
  (|cs|+1)^b-1$.
  %So $\numinterpret{(s,t)}^\polmain_{cs}$ is
  %the number with bit representation $b_0\dots b_N$ where $b_i = 1$
  %if{f} $s \mapsto i$, if{f} $t \not\mapsto i$ (with $b_0$ the most
  %significant digit).
\item $\Sigma^\polmain$ consists of the defined symbols introduced in
  $R^\polmain$, which are those in Lemma~\ref{lem:mainmodule}
  with $\symb{seed}^1_{\mathtt{lin}},\symb{zero}_{\mathtt{lin}},
  \symb{pred}^1_{\mathtt{lin}}$ and $\symb{eqLen}$ replaced by
  $\symb{seed}_{\mathtt{base}},\symb{zero}_{\mathtt{base}}$,
  $\symb{pred}_{\mathtt{base}}$ and $\symb{eqBase}$ respectively,
  along with the following supporting rules:
\[
\begin{array}{rcl}
\symb{seed}_{\mathtt{base}}\ cs & \arrz &
  (\nul\cons\dots\cons\nul\cons\nil, cs, \dots, cs) \\
  \multicolumn{3}{r}{\llbracket
  \text{with}\ |\nul\cons\dots\cons\nul\cons\nil| = a-1
%  \text{and}\ b\ \text{copies of}\ cs
  \rrbracket}
\\
\symb{zero}_{\mathtt{base}}\ cs\ (\nil,\dots,\nil) & \arrz & \strue \\
\symb{zero}_{\mathtt{base}}\ cs\ (xs_0,\dots,xs_{i-1},y\cons ys,\nil,
  \dots,\nil) & \arrz & \sfalse
  \hfill
  \llbracket\text{for}\ 0 \leq i \leq b\rrbracket \\
\symb{pred}_{\mathtt{base}}\ cs\ (\nil,\dots,\nil) & \arrz &
  (\nil,\dots,\nil) \\
\symb{pred}_{\mathtt{base}}\ (c\cons zs)\ (xs_0,\dots,xs_{i-1},y
  \cons ys,\nil,\dots,\nil) & \arrz & (xs_0,\dots,xs_{i-1},ys,c\cons
  zs,\dots,c\cons zs) \\
& & \hfill \llbracket\text{for}\ 0 \leq i \leq b\rrbracket \\
\end{array}
\]
\[
\begin{array}{rcl}
\symb{eqBase}\ (\nil,\dots,\nil)\ (\nil,\dots,\nil) & \arrz & \strue
  \\
\symb{eqBase}\ (xs_0,\dots,xs_{i-1},y\cons ys,\nil,\dots,\nil)\ 
  (zs_0,\dots,zs_{i-1},\nil,\nil,\dots,\nil) & \arrz & \sfalse \\
\symb{eqBase}\ (xs_0,\dots,xs_{i-1},\nil,\nil,\dots,\nil)\ 
  (zs_0,\dots,zs_{i-1},y\cons ys,\nil,\dots,\nil) & \arrz & \sfalse
  \phantom{ABC} \\
\symb{eqBase}\ (xs_0,\dots,xs_{i-1},y\cons ys,\nil,\dots,\nil)\ 
 (zs_0,\dots,zs_{i-1},n\cons ns,\nil,\dots,\nil) & \arrz \\
\multicolumn{3}{r}{
  \symb{eqBase}\ (xs_0,\dots,xs_{i-1},ys,\nil,\dots,\nil)\ 
  (zs_0,\dots,zs_{i-1},ns,\nil,\dots,\nil)}
\end{array}
\]
\end{itemize}
This module functions as the ones from
Lemmas~\ref{lem:mainmodule} and~\ref{lem:countpairs}.  Note that in
the rules for $\symb{pred}_{\mathtt{base}}$, we expanded the
variable $cs$ representing the input list to keep these rules
product-\linebreak cons-free.
By using $c\cons zs$, the list is guaranteed to be normalized (and
non-empty).
\end{proof}

Thus, combining Lemmas~\ref{lem:countaptrsgood}
and~\ref{lem:expQcount} with the rules of Figure~\ref{fig:TM},
we obtain:

\begin{corollary}\label{cor:aptrsatleast}
Any decision problem in $\exptime{K}$ is accepted by a
$K^{\text{th}}$-order product-cons-free APTRS.
%Here, 
%$
%\exptime{K} \triangleq \bigcup_{a,b \in \nats}
%\timecomp{\textrm{exp}_2^{K}(an^b)}
%$.
\end{corollary}

%\begin{proof}
%By Lemmas~\ref{lem:countaptrsgood},
%\ref{lem:expQcount}, and the ATRS of Figure~\ref{fig:TM}.
%\end{proof}

By standard results, $\etime{K} \subsetneq \exptime{K}$ for all $K
\geq 1$, hence the addition of pairing materially increases
expressivity. Conversely, we have:

\begin{theorem}\label{thm:aptrsatmost}
Any set accepted by a $K^{\text{th}}$-order
product-cons-free APTRS is in $\exptime{K}$.
\end{theorem}

\begin{proof}
Following the proof of Lemma~\ref{lem:algorithmcomplexity},
the complexity of Algorithm~\ref{alg:main} is polynomial in the
cardinality of the largest $\interpret{\atype}$ used.
The result follows by letting $\interpret{\asort_1 \times \dots \times 
\asort_n}$ contain subsets of $\B^n$---which we can do because the
only pairs occurring in a reduction are data.

Formally, let $(\F,\Rules)$ be a product-cons-free APTRS.
We first prove that any pair occurring in a reduction $s \arrr{\Rules}
t$ with $s$ basic, is a data term.
Let a term $s$ be
\emph{product-$\B$-safe} if $s$ is $\B$-safe and
$s \suptermeq (s_1,s_2)$ implies $(s_1,s_2) \in \Data$ for
all $s_1,s_2$.  We observe:
\emph{(**) if $s$ is product-$\B$-safe and $s \arr{\Rules} t$, then
$t$ is product-$\B$-safe.}
$\B$-safety of $t$ follows by Lemma~\ref{lem:safetypreserve} and for
any $t \suptermeq (t_1,t_2)$:
if not $s \suptermeq (t_1,t_2)$, then
there are $\ell \arrz r \in \Rules$, substitution $\gamma$
and $r_1,r_2$ such that $s \suptermeq \ell\gamma$ and $r \suptermeq
(r_1,r_2)$ and $(t_1,t_2) = (r_1,r_2)\gamma$.
By definition
of product-cons-free, each $r_i$ is a pair---so $r_i\gamma$
is data by induction
on the size of $r$---or has the form $\identifier{c}\ 
\vec{s}$---so
$r_i\gamma$ is a data term by $\B$-safety of $t$---or is a variable
$x$ such that $\gamma(x) \in \Data$ by product-$\B$-safety of $s$.
Thus, $(r_1,r_2)$ is a pair of two data terms, and therefore data
itself.

Thus, we can safely assume that the only product types that occur
have type order $0$, and remove constructors or defined symbols
using higher order product types.

Next, we adapt Algorithm~\ref{alg:main}.  We denote all
types of order $0$ as $\asort_1 \times \dots \times \asort_n$ (ignoring
bracketing) and let $\interpret{\asort_1 \times \dots \times
\asort_n} = \P(\{ (s_1,\dots,s_n) \mid s_i \in \B \wedge \vdash s_i :
\asort_i$ for all $i\})$.  Otherwise, the algorithm
is unaltered.
Let $b$ be the longest length of any product type occurring in $\F$.
As $\P(\B^b)$ has cardinality $2^{N^b}$, the reasoning in
Lemma~\ref{lem:algorithmcomplexity} gives
$\card(\interpret{\atype}) \leq \exp_2^{K+1}(d^K \cdot N^b)$ for a
type of order $k$, which results in $\timecomp{\exp_2^K(a \cdot n^b)}$
for the algorithm.
Lemma~\ref{lem:labeledequiv} goes through unmodified,
Lemma~\ref{lem:terminating} goes through if we define $(s,t)$ to be
computable if both $s$ and $t$ are, and
Lemma~\ref{lem:algorithmsoundcomplete} by using
product-$\B$-safety instead of $\B$-safety in case (B).
\end{proof}

Thus we obtain:

\begin{corollary}
A decision problem $X$ is in $\exptime{K}$ if and only if there is a
$K^{\text{th}}$-order product-cons-free APTRS which accepts $X$%:
%the class of product-cons-free APTRSs with order $K$
%\emph{characterizes} $\exptime{K}$.
.
\end{corollary}

\section{Altering ATRSs}

As demonstrated in \secshort\ref{sec:pairing}, the expressivity of
cons-free term rewriting is highly sensitive in the presence of
non-determinism:
minor syntactical changes have the potential to significantly affect
expressivity.  In this section, we briefly discuss three other groups
of changes.

\subsection{Strategy}

In moving from functional programs to term rewriting, we diverge
from Jones' work in two major ways: by allowing
non-deterministic choice, and by not imposing a reduction strategy.
Jones' language in~\cite{jon:01} employs \emph{call-by-value}
reduction.  A close parallel in term rewriting is to consider
\emph{innermost} reductions, where a step $\ell\gamma
\arr{\Rules} r\gamma$ may only be taken if all strict subterms
of $\ell\gamma$ are in normal form.  Based on results by Jones and
Bonfante, and our own %(not yet published)
work on call-by-value
programs, we conjecture the following claims: %characterizations:

\begin{enumerate}
\item\label{it:confluent}
  \emph{confluent} cons-free ATRSs of order $K$, with
  innermost reduction, characterize \\ $\exptime{K-1}$;
  here, $\exptime{0} = \ptime$, the sets decidable in polynomial time
\item\label{it:nondetfirst}
  cons-free ATRSs of order $1$, with innermost reduction,
  characterize $\ptime$
\item\label{it:nondethigher}
  cons-free ATRSs of order $> 1$, with innermost reduction,
  characterize $\elementary$
\end{enumerate}

(\ref{it:confluent}) is a direct translation of Jones' result on time
complexity from~\cite{jon:01} to innermost rewriting.
(\ref{it:nondetfirst}) translates Bonfante's
result~\cite{DBLP:conf/amast/Bonfante06}, which states that adding a
non-deterministic choice operator to Jones' language does not increase
expressivity in the first-order case.
(\ref{it:nondethigher}) is our own result, presented (again for
call-by-value programs) in~\cite{kop:sim:17}.  The reason for the
explosion is that we can define a similar counting module as the one
for pairing in Lemma~\ref{lem:countpairs}.

Each result can be proved with an argument similar to the one in this
paper:
for one direction, a TM simulation with counting modules;
for the other, an algorithm to evaluate the cons-free program.
While the original results admit pairing, this
adds no expressivity as the simulations can be specified
without pairs.
%We believe that the proof changes needed to use innermost rather than
%call-by-value reduction are straightforward, by including a special
%symbol $\bot$ in the set $\B$ to represent a normal form which is not
%a data term. 
We believe that the proof is easily changed to accommodate innermost
over call-by-value reduction, but have not done this formally.

\medskip
Alternatively, we may consider \emph{outermost} reductions steps,
where rules are always applied at the highest possible position in a
term.
Outermost reductions are semi-outermost, but may behave
differently in the presence of overlapping rules; for example, given
rules $\symb{f}\ \nul \arrz \strue$ and $\symb{f}\ x \arrz \sfalse$,
an outermost evaluation would have to reduce $\symb{f}\ 
(\nul+\nul)$ to $\sfalse$, while in a semi-outermost evaluation we could
also have $\symb{f}\ (\nul+\nul) \arr{\Rules} \symb{f}\ \nul \arr{
\Rules} \strue$.
We note that
%\begin{itemize}
%\item
  the ATRS from Figure~\ref{fig:TM} and all counting modules
  evaluate as expected %when
  using outermost reduction %;
and that
%\item
  Theorem~\ref{thm:algorithm} does not consider evaluation
  strategy.
%  in the proof of Lemma~\ref{lem:algorithmsoundcomplete}, the
%  $\arrr{\Ruleslab}$ reduction constructed from an algorithm
%  evaluation is outermost: subterms are only evaluated to
%  match a pattern $\ell_i$.
%\end{itemize}
%
This gives:

\begin{corollary}\label{cor:exptime}
A decision problem $X$ is in $\etime{K}$ if and only if there is a
$K^{\text{th}}$-order cons-free ATRS with outermost reduction which
accepts $X$.
\end{corollary}

\subsection{Constructor ATRSs and left-linearity}\label{sec:llcatrs}

%Before even defining cons-freeness, we limited interest to
%\emph{left-linear constructor ATRSs}.  
%One might wonder if dropping
%either restriction would still give an interesting class of decision
%problems.
Recall that we have exclusively considered \emph{left-linear
constructor ATRSs}.  One may wonder whether these restrictions can be
dropped.

The answer, however, is no.  In the case of constructor ATRSs, this is
easy to see: if we do not limit interest to constructor ATRSs---so if,
in a rule $\apps{\identifier{f}}{\ell_1}{\ell_k} \arrz r$ the terms
$\ell_i$ are not required to be patterns---then ``cons-free'' becomes
meaningless, as we could simply let $\Defineds := \F$.  Thus, we would
obtain a Turing-complete language already for first-order ATRSs.

Removing the requirement of left-linearity similarly provides full
Turing-completeness%, although this requires more effort to see.  It
.  This
is demonstrated by the first-order cons-free ATRS in
Figure~\ref{fig:nonlinear} which simulates an arbitrary TM
on input alphabet $I = \{0,1\}$.  A tape $x_0 \dots x_n\blank
\blank\dots$ with the reading head at position $i$ is represented by
three parameters: $x_{i-1}\symb{::}\dots\symb{::} x_0$ and $x_i$ and
$x_{i+1}\symb{::}\dots\symb{::} x_n$.  Here, the ``list
constructor'' $\symb{::}$ is a \emph{defined symbol}, ensured by
a rule which never fires.  To split a ``list'' into a head and tail,
the ATRS non-deterministically generates a \emph{new} head and tail
using two calls to $\symb{rndtape}$ (whose only shared reducts are
fully evaluated ``lists''), and uses a
non-left-linear rule
to compare their combination to the original ``list''.

%An additional equality
%check guarantees that the new tail has been fully evaluated to its
%``list'' normal form.

\begin{figure}[htb]
%\vspace{-12pt}
\[
\begin{array}{rclcrclcrcl}
\symb{rndtape}\ x & \arrz & \nil & \quad &
\symb{rnd} & \arrz & \nul & \\
\symb{rndtape}\ x & \arrz &  \symb{rnd::rndtape}\ x & \quad &
\symb{rnd} & \arrz & \one \\
\bot\symb{::}t & \arrz & t & \quad &
\symb{rnd} & \arrz & \symb{B} \\
\end{array}
\]
\[
\begin{array}{rcl}
\symb{translate}\ (\nul\cons xs) & \arrz & \nul~\symb{::}~
  (\symb{translate}\ xs) \\
\symb{translate}\ (\one\cons xs) & \arrz & \one~\symb{::}~
  (\symb{translate}\ xs) \\
\symb{translate}\ \nil & \arrz & \symb{B}~\symb{::}~
  (\symb{translate}\ \nil) \\
\symb{translate}\ \nil & \arrz & \nil \\
\symb{equal}\ xl\ xl & \arrz & \strue \\
\end{array}
\]
\[
\begin{array}{rcl}
\symb{start}\ cs & \arrz & \symb{run}\ \symb{start}\ \nil\ \symb{B}\ 
  (\symb{translate}\ cs) \vspace{-6pt} \\
\symb{run}\ \unknown{s}\ xl\ \unknown{r}\ yl & \arrz &
  \symb{shift}\ \unknown{t}\ xl\ \unknown{w}\ yl\ \unknown{d})\ 
  \ \ 
  \llbracket\text{for every transition}\ \transition{\unknown{s}}{
  \unknown{r}}{\unknown{w}}{\unknown{d}}{\unknown{t}}\rrbracket \\
\symb{shift}\ s\ xl\ c\ yl\ d & \arrz & \symb{shift}_1\ s\ xl\ c\ yl\ 
  d\ \symb{rnd}\ (\symb{rndtape}\ \nul)\ (\symb{rndtape}\ \one) \\
\symb{shift}_1\ s\ xl\ c\ yl\ d\ \unknown{b}\ t\ t & \arrz &
  \symb{shift}_2\ s\ xl\ c\ yl\ d\ \unknown{b}\ t\ \ \llbracket
  \text{for every}\ \unknown{b} \in \{\symb{O},\symb{I},\symb{B}\}
  \rrbracket \\
\symb{shift}_2\ s\ xl\ c\ yl\ \symb{R}\ z\ t & \arrz &
  \symb{shift}_3\ s\ (c~\symb{::}~xl)\ z\ t\ 
  (\symb{equal}\ yl\ (z~\symb{::}~t)) \\
\symb{shift}_2\ s\ xl\ c\ yl\ \symb{L}\ z\ t & \arrz &
  \symb{shift}_3\ s\ t\ z\ (c~\symb{::}~yl)\ (\symb{equal}\ xl\ 
  (z~\symb{::}~t)) \\
\symb{shift}_3\ s\ xl\ c\ yl\ \strue & \arrz &
  \symb{run}\ s\ xl\ c\ yl \\
\end{array}
\]
\caption{A first-order non-left-linear ATRS that simulates a given
Turing machine}
\label{fig:nonlinear}
\end{figure}

\subsection{Variable binders}

A feature present in many styles of higher-order term rewriting is
\emph{$\lambda$-abstraction}; e.g., a construction such as $\lambda x.
\identifier{f}\ x$.  Depending on the implementation, admitting
$\lambda$-abstraction in cons-free ATRSs may blow up expressivity, or
not affect it at all.

%For this brief discussion, we assume a basic understanding of
%the $\lambda$-calculus.
%$\lambda$-abstraction, variable binding and $\alpha$-conversion (which
%is done in the same way as in the $\lambda$-calculus).

First, consider ATRSs with $\lambda$-abstractions used only in the
right-hand sides of rules.  Then all abstractions can
be removed by introducing fresh function symbols, e.g., by replacing
a rule
$\symb{f}\ (\symb{c}\ y) \arrz \symb{g}\ 
(\lambda x.\symb{h}\ x\ y)$
%$\apps{\identifier{f}}{\ell_1}{\ell_k} \arrz C[\lambda x.s]$
%with $\Var(\ell_1) \cup \dots \cup \Var(\ell_k) = \{ y_1,\dots,y_n \}$
by the two rules
$\symb{f}\ (\symb{c}\ y) \arrz \symb{g}\ (
\symb{f}_{\mathtt{help}}\ y)$ and $\symb{f}_{\mathtt{help}}\ y\ x \arrz
\symb{h}\ x\ y$
%$\apps{\identifier{f}}{\ell_1}{\ell_k} \arrz C[
%\apps{\identifier{h}}{y_1}{y_n}]$ and $\apps{\identifier{h}}{y_1}{y_n}
%\arrz s$
(where $\identifier{f}_{\mathtt{help}}$ is a fresh symbol).
%Thus, this feature does not add any expressivity.%
%\footnote{Note that this alteration \emph{does} affect some
%  reductions: for example, a term $\identifier{f}\ 
%  (\lambda x.\symb{0} + x)$ may be reduced to $\identifier{f}\ 
%  (\lambda x.x)$ while $\identifier{f}\ \identifier{h}$ cannot be so
%  reduced, even in the presence of a rule $\identifier{h}\ x \arrz
%  \symb{0} + x$.  However, in an evaluation from a basic term to a
%  data term, such reduction steps can always be postponed, as they
%  are not outermost.}
Since the normal forms of basic terms are not affected by this
change, this feature adds no expressivity.

Second, some variations of higher-order term rewriting require that
function symbols are always assigned to as many arguments as possible;
abstractions are the \emph{only} terms of functional type.
Clearly, this does not increase expressivity as it merely limits the
number of programs (with $\lambda$-abstraction) that we can specify.
Nor does it lower expressivity: the results in this paper
go through in such a formalism, as demonstrated in~\cite{kop:sim:16}.
It does, however, require some changes to the definition of a counting
module.
%However, it does not lower expressivity either: the program of
%Figure~\ref{fig:TM} satisfies this restriction, and the definition of
%counting modules and Lemma~\ref{lem:expQcount} can be altered
%accordingly.  We did this in~\cite{kop:sim:16}.

Finally, if abstractions are allowed in the \emph{left}-hand sides of
rules, then the same problem arises as in Lemma~\ref{lem:countpairs}:
we can force a partial evaluation,
and use this to define $(\lambda n.\exp_2^K(n))$-counting modules for
arbitrarily high $K$ without increasing type orders.  This is because
a rule such as $\identifier{f}\ (\lambda x.Z)$ matches a term
$\identifier{f}\ (\lambda x.\symb{0})$, but does \emph{not} match
$\identifier{f}\ (\lambda x.\symb{g}\ x\ \symb{0})$ because of how
substitution works in the presence of binders.  A full exposition of
this issue %---and the definition of a second-order counting module
%similar to the one in Lemma~\ref{lem:countpairs}---
would require a
more complete definition of higher-order term rewriting with
$\lambda$-abstraction, so is left as an exercise to interested
readers.
%
%We do not know whether it is possible to obtain comparable results in
%a higher-order rewriting setting where rewriting is done \emph{modulo}
%$\beta\eta$-reduction, if a restriction such as \emph{fully extended
%rules} is imposed.  This is left for future work.
A restriction such as \emph{fully extended rules} may be used to
bypass this issue; we leave this question to future work.

\section{Conclusions}\label{sec:conclusion}

We have studied the expressive power of cons-free
higher-order term rewriting, and seen that
restricting data order results in
characterizations of different classes.
We have shown
that pairing dramatically increases this expressive
power---and how this can be avoided by using additional
restrictions---and we have briefly
discussed the  effect of other syntactical changes.
The main results are displayed in Figure~\ref{fig:conclusion}.
%along with the corresponding result in~\cite{jon:01} for comparison.

\begin{figure}[htb]
%\begin{center}
%\textbf{Systems \emph{P} with type order $K$ characterize the class \emph{C}}
%\end{center}
%
%\vspace{5pt}
  \def\arraystretch{1.2}
\begin{tabular}{m{6cm} m{3.5cm} m{4cm}}
\hline
  \centering \textbf{P}
& \centering \textbf{C}
& \\
\hline
& & \\[-0.7em]
\parbox{5cm}{confluent cons-free ATRSs with call-by-value reduction} &
$\exptime{K-1}$ &
(translated from~\cite{jon:01}) \\
& & \\[-0.5em]
%\hline
%\vphantom{$2^{2^{2^o}}_{2_2}$}%
cons-free ATRSs &
$\etime{K}$ &
(Corollary~\ref{cor:main}) \\
& & \\[-0.5em]
%\hline
%\vphantom{$2^{2^{2^o}}_{2_2}$}%
product-cons-free APTRSs &
$\exptime{K}$ &
(Corollary~\ref{cor:exptime}) \\
& & \\[-0.5em]
%\hline
%\vphantom{$2^{2^2}_{2_2}$}%
cons-free APTRSs (so with pairing) &
$\geq \elementary$ &
(Corollary~\ref{cor:atleastelem}) \\
& & \\[-0.7em]
\hline
\end{tabular}

%\begin{tabular}{|c|c|c|c|c}
%& \textbf{Order 1} & \textbf{Order 2} & \textbf{Order 3} & \dots \\
%\hline
%\textbf{Confluent cons-free} & & & \\
%\textbf{call-by-value} &
%  $\exptime{0}$ &
%  $\exptime{1}$ &
%  $\exptime{2}$ & \dots \\
%\textbf{rewriting (from~\cite{jon:01})} & & & \\
%\hline
%\vphantom{$2^{2^{2^2}}_{2_{2_{2_2}}}$} \textbf{cons-free rewriting} &
%  $\etime{1}$ &
%  $\etime{2}$ &
%  $\etime{3}$ &
%  \dots \\
%\hline
%\textbf{cons-free rewriting} &
%  \multirow{2}{*}{$\elementary$} &
%  \multirow{2}{*}{$\elementary$} &
%  \multirow{2}{*}{$\elementary$} &
%  \multirow{2}{*}{\dots} \\
%\textbf{with pairing} & & & \\
%\hline
%\textbf{product-cons-free re-} &
%  \multirow{2}{*}{$\exptime{1}$} &
%  \multirow{2}{*}{$\exptime{2}$} &
%  \multirow{2}{*}{$\exptime{3}$} &
%  \multirow{2}{*}{\dots} \\
%\textbf{writing with pairing} & & & \\
%\hline
%\end{tabular}
%\caption{Expressive power of cons-free applicative term rewriting}
\caption{Overview: systems \textbf{P} with type order $K$ characterize
the class \textbf{C}.}
\label{fig:conclusion}
\end{figure}

\subsection{Future work}

We see two major, natural lines of further inquiry, that we believe
will also be of significant interest in the general---non-rewriting
related---area of implicit complexity.
Namely (I), the imposition of further
restrictions, either on rule formation, reduction strategy or both
that, combined with higher-order rewriting will yield characterization
of  \emph{non}-deterministic classes such as
$\textrm{NP}$, or of sub-linear time classes like $\textrm{LOGTIME}$.
And (II), additions of %a notion of
\emph{output}. While cons-freeness does not naturally lend itself to
producing output, it is common in implicit complexity to investigate
characterizations of sets of computable \emph{functions}, e.g. the
polytime-computable functions on integers, rather than
decidable sets.
This could for instance be done by allowing the production
of constructors of specific types.

\vspace{-1pt}

%\subparagraph*{Acknowledgements}
%
%The authors thank \dots

\bibliography{references}

\begin{thebibliography}{10}

\bibitem{DBLP:conf/aplas/AvanziniEM12}
M.~Avanzini, N.~Eguchi, and G.~Moser.
\newblock A new order-theoretic characterisation of the polytime computable
  functions.
\newblock In {\em \proc APLAS}, volume 7705 of {\em LNCS}, pages 280--295,
  2012.

\bibitem{DBLP:conf/rta/AvanziniM10}
M.~Avanzini and G.~Moser.
\newblock Closing the gap between runtime complexity and polytime
  computability.
\newblock In {\em \proc RTA}, volume~6 of {\em LIPIcs}, pages 33--48, 2010.

\bibitem{DBLP:journals/corr/AvanziniM13}
M.~Avanzini and G.~Moser.
\newblock Polynomial path orders.
\newblock {\em Logical Methods in Computer Science}, 9(4), 2013.

\bibitem{DBLP:conf/tlca/Baillot07}
P.~Baillot.
\newblock From proof-nets to linear logic type systems for polynomial time
  computing.
\newblock In {\em \proc TLCA}, volume 4583 of {\em LNCS}, pages 2--7, 2007.

\bibitem{baillot_et_al:LIPIcs:2012:3664}
P.~Baillot and U.~Dal Lago.
\newblock {Higher-Order Interpretations and Program Complexity}.
\newblock In {\em \proc CSL}, volume~16 of {\em LIPIcs}, pages 62--76, 2012.

\bibitem{DBLP:conf/esop/BaillotGM10}
Patrick Baillot, Marco Gaboardi, and Virgile Mogbil.
\newblock A polytime functional language from light linear logic.
\newblock In {\em \proc ESOP}, volume 6012 of {\em LNCS}, pages 104--124, 2010.

\bibitem{DBLP:journals/cc/BellantoniC92}
S.~Bellantoni and S.~Cook.
\newblock A new recursion-theoretic characterization of the polytime functions.
\newblock {\em Computational Complexity}, 2:97--110, 1992.

\bibitem{DBLP:journals/apal/BellantoniNS00}
S.~Bellantoni, K.~Niggl, and H.~Schwichtenberg.
\newblock Higher type recursion, ramification and polynomial time.
\newblock {\em Annals of Pure and Applied Logic}, 104(1--3):17--30, 2000.

\bibitem{bla:jou:rub:08}
F.~Blanqui, J.~Jouannaud, and A.~Rubio.
\newblock The computability path ordering: The end of a quest.
\newblock In {\em \proc CSL}, volume 5213 of {\em LNCS}, pages 1--14, 2008.

\bibitem{DBLP:conf/amast/Bonfante06}
G.~Bonfante.
\newblock Some programming languages for logspace and ptime.
\newblock In {\em \proc AMAST}, volume 4019 of {\em LNCS}, pages 66--80, 2006.

\bibitem{car:sim:14}
D.~de~Carvalho and J.~Simonsen.
\newblock An implicit characterization of the polynomial-time decidable sets by
  cons-free rewriting.
\newblock In {\em \proc RTA-TLCA}, volume 8560 of {\em LNCS}, pages 179--193,
  2014.

\bibitem{HELLA2006197}
Lauri Hella and José~María Turull-Torres.
\newblock Computing queries with higher-order logics.
\newblock {\em Theoretical Computer Science}, 355(2):197 -- 214, 2006.

\bibitem{Hofmann:hab}
M.~Hofmann.
\newblock Type systems for polynomial-time computation, 1999.
\newblock Habilitations\-schrift.

\bibitem{Jones:CompComp}
N.~Jones.
\newblock {\em Computability and Complexity from a Programming Perspective}.
\newblock MIT Press, 1997.

\bibitem{jon:01}
N.~Jones.
\newblock The expressive power of higher-order types or, life without {CONS}.
\newblock {\em Journal of Functional Programming}, 11(1):55--94, 2001.

\bibitem{kop:sim:16}
C.~Kop and J.~Simonsen.
\newblock Complexity hierarchies and higher-order cons-free rewriting.
\newblock In {\em \proc FSCD}, volume~52 of {\em LIPIcs}, pages 23:1--23:18,
  2016.

\bibitem{kop:sim:17}
C.~Kop and J.~Simonsen.
\newblock The power of non-determinism in higher-order implicit complexity.
\newblock In {\em \proc ESOP}, volume 10201 of {\em LNCS}, pages 668--695,
  2017.

\bibitem{DBLP:journals/tcs/KristiansenN04}
L.~Kristiansen and K.~Niggl.
\newblock On the computational complexity of imperative programming languages.
\newblock {\em Theoretical Computer Science}, 318(1--2):139--161, 2004.

\bibitem{KUPER199333}
Gabriel~M. Kuper and Moshe~Y. Vardi.
\newblock On the complexity of queries in the logical data model.
\newblock {\em Theoretical Computer Science}, 116(1):33 -- 57, 1993.

\bibitem{Papadimitriou:complexity}
C.~Papadimitriou.
\newblock {\em Computational Complexity}.
\newblock Addison-Wesley, 1994.

\bibitem{Sipser:comp}
M.~Sipser.
\newblock {\em Introduction to the Theory of Computation}.
\newblock Thomson Course Technology, 2006.

\bibitem{T03_R}
F.~van Raamsdonk.
\newblock Higher-order rewriting.
\newblock In {\em Term Rewriting Systems}, {C}hapter~11, pages 588--667.
  Cambridge University Press, 2003.

\end{thebibliography}
\bibliographystyle{plain}

\end{document}